\documentclass[12pt]{article}

\usepackage{amsmath,amssymb,amsthm,amscd,a4wide,upref,comment}

\usepackage[enableskew]{youngtab}
\usepackage{ytableau}
\usepackage{enumerate}

\usepackage{hyperref}
\hypersetup{colorlinks,citecolor=blue,filecolor=black,linkcolor=blue,urlcolor=blue}

\usepackage{lmodern}     
\usepackage[T1]{fontenc}

\allowdisplaybreaks

 \newtheorem{thm}{Theorem}[section]

 \newtheorem{cor}[thm]{Corollary}
 \newtheorem{lem}[thm]{Lemma}
 \newtheorem{prop}[thm]{Proposition}
 
 \newtheorem{defn}[thm]{Definition}
  
  \newtheorem{defn-thm}[thm]{Definition-Theorem}
  \newtheorem{ex}[thm]{Example}

  \theoremstyle{remark}
\newtheorem{rem}[thm]{Remark}

\numberwithin{equation}{section}

\newcommand{\row}{\text{row}}
\newcommand{\col}{\text{col}}
\newcommand{\de}{\text{deg}}

\DeclareFontFamily{U}{mathx}{}
\DeclareFontShape{U}{mathx}{m}{n}{<-> mathx10}{}
\DeclareSymbolFont{mathx}{U}{mathx}{m}{n}
\DeclareMathAccent{\widehat}{0}{mathx}{"70}
\DeclareMathAccent{\widecheck}{0}{mathx}{"71}

\newcommand{\non}{\nonumber}
\newcommand{\wt}{\widetilde}

\newcommand{\la}{\lambda}

\newcommand{\ts}{\,}

\newcommand{\tss}{\hspace{1pt}}

\newcommand{\CC}{\mathbb{C}\tss}

\newcommand{\Wc}{\mathcal{W}}

\newcommand{\gl}{\mathfrak{gl}}

\newcommand{\g}{\mathfrak{g}}

\newcommand{\sll}{\mathfrak{sl}}
\newcommand{\agot}{\mathfrak{a}}

\newcommand{\Fand}{\qquad\text{and}\qquad}

\newcommand{\bal}{\begin{aligned}}
\newcommand{\eal}{\end{aligned}}
\newcommand{\beq}{\begin{equation}}
\newcommand{\eeq}{\end{equation}}
\newcommand{\ben}{\begin{equation*}}
\newcommand{\een}{\end{equation*}}

\newcommand{\bpf}{\begin{proof}}
\newcommand{\epf}{\end{proof}}

\def\beql#1{\begin{equation}\label{#1}}

\begin{document}

\title{\Large\bf Generalized finite and affine $W$-algebras in type $A$}

\author{{Dong Jun Choi,\quad Alexander Molev\quad and\quad Uhi Rinn Suh}}

\date{} 
\maketitle


\begin{abstract}
We construct
a new family of affine $W$-algebras $W^k(\lambda,\mu)$
parameterized by partitions $\lambda$ and $\mu$
associated with
the centralizers of nilpotent elements in $\gl_N$.
The new family unifies a few known classes of $W$-algebras.
In particular, for
the column-partition $\lambda$ we recover the affine $W$-algebras $W^k(\gl_N,f)$
of Kac, Roan and Wakimoto,
associated with nilpotent elements $f\in\gl_N$ of type $\mu$.
Our construction is based on
a version of the BRST complex of the quantum
Drinfeld--Sokolov reduction. We show that the application of the Zhu functor
to the vertex algebras $W^k(\lambda,\mu)$
yields a family of generalized finite $W$-algebras $U(\lambda,\mu)$
which we also describe independently as associative algebras.



\end{abstract}

\maketitle

\section{Introduction}
\label{sec:int}
\subsection{\texorpdfstring{$W$}{W}-algebras} \label{subsec:W-}
The appearance of {\em affine $W$-algebras}
was motivated by physics, originating in the work of Zamolodchikov~\cite{z:ie};
he extended the
Virasoro algebra by adding fields of higher conformal dimension thus discovering
the $W_3$ algebra associated with the Lie algebra $\sll_3$.
The $W_N$-algebras corresponding to $\sll_N$ were then constructed by Fateev and Lukyanov~\cite{fl:mt}.
Due to their significance in the conformal field theory, the $W$-algebra symmetries
were since extensively studied in the physics literature; see e.g. \cite{bs:ws} for a review.

A conceptual mathematical definition of the $W$-algebras in the context of the vertex algebra theory
is due to Feigin and Frenkel~\cite{ff:qd} who introduced them via
the quantized Drinfeld--Sokolov
reduction. According to this definition, the $W$-algebra $W^{k}(\g)$ at the level $k\in\CC$
corresponding
to the simple Lie algebra $\g$ is the zeroth cohomology of the BRST complex of the
quantum Drinfeld--Sokolov reduction; see \cite{a:iw} and \cite[Ch.~15]{fb:va} for detailed expositions. 

The theory was further developed by Kac, Roan and Wakimoto \cite{krw:qr}
introducing a wider class of $W$-algebras $W^k(\g,f)$ associated with simple
Lie algebras and superalgebras $\g$ and nilpotent elements $f\in\g$.
For a principal nilpotent $f$ the algebra $W^k(\g,f)$ coincides with $W^k(\g)$ and for the trivial nilpotent $f$ the corresponding $W$-algebra $W^k(\g,f)$ is the affine vertex algebra $V^k(\g).$ In other words, for given Lie superalgebra $\g$, the family of $W$-algebras interpolate the principal $W$-algebra and affine vertex algebra.
In the principal case, the investigation of the
structural theory and representations of $W$-algebras led
Frenkel, Kac and Wakimoto~\cite{fkw:cf} to
remarkable conjectures
concerning their rationality and the existence and description of
modular invariant representations. The $W$-algebra representation theory
was further developed by Arakawa in \cite{a:rt} and \cite{a:rw}, where both
conjectures were proved.

In the limit $k\to\infty$ the algebra $W^k(\g)$ becomes the
{\em classical $W$-algebra} $W(\g)$ which is a commutative algebra equipped with a Poisson bracket.
Its definition goes back to Drinfeld and Sokolov~\cite{ds:la}; they used the algebras $W(\g)$ to introduce
equations of the KdV type for arbitrary simple Lie algebras.
In a more recent work by De Sole, Kac and Valeri~\cite{dskv:cw}
the construction of Drinfeld and Sokolov (generalized to an arbitrary nilpotent
element $f$) was described in the framework
of Poisson vertex algebras. This description
was then applied to construct integrable hierarchies
of bi-Hamiltonian equations.

According to the Feigin--Frenkel duality property~\cite{ff:qd},
another way to recover the classical $W$-algebra $W({}^L\g)$
associated with the Langlands dual Lie algebra ${}^L\g$ is to take
the {\em critical level} for $W^k(\g)$; i.e., to evaluate $k$ at the
negative of the dual Coxeter number of $\g$.
Due to the celebrated theorem of Feigin and Frenkel~\cite{ff:ak},
the algebra $W({}^L\g)$ is isomorphic to the center of the
affine vertex algebra at the critical level corresponding to $\g$.

An independent theory of {\em finite $W$-algebras} was conceived in the work
of Kostant~\cite{k:wv}. More recent approaches were developed
by Premet~\cite{p:st} and Gan and Ginzburg~\cite{gg:qs}.
The connection between the affine and finite $W$-algebras was established by
De~Sole and Kac~\cite{dk:fa}, and (in the principal nilpotent case) by
Arakawa~\cite{a:rt}. They showed that
the application of the {\em Zhu functor} originated in \cite{z:mi} to
the affine $W$-algebra $W^{\tss k}(\g,f)$ associated with a nilpotent element $f\in\g$
yields the finite $W$-algebra corresponding to $\g$ and $f$.

\subsection{Generalized \texorpdfstring{$W$}{W}-algebras}

It is well-known that some non-semisimple Lie algebras can share
certain classical properties of their semisimple counterparts. Amongst such examples
are centralizers $\agot$ of nilpotent elements in simple Lie algebras. In particular,
the remarkable {\em Premet conjecture} states that the subalgebra
of $\agot$-invariants in the symmetric algebra $S(\agot)$
is a free polynomial algebra; see \cite{ppy:si}. Although it does not
hold in full generality \cite{y:cp}, the conjecture inspired further research
into these Lie algebras and associated objects. The centers of the universal
enveloping algebras $U(\agot)$ in type $A$ were constructed explicitly in \cite{bb:ei} and
generators of the classical $W$-algebras associated with the centralizers were produced
in \cite{mr:cw}.

The main motivation of our work comes from the results of Arakawa and Premet~\cite{ap:qm}.
They found a remarkable way to use affine $W$-algebras to describe the
centers of the universal affine vertex algebras
at the critical level associated with the centralizers $\agot$, thus
obtaining a version of the Feigin--Frenkel theorem for centralizers.
This description was further applied  to quantize the Mishchenko--Fomenko
subalgebras in $U(\agot)$.
The center at the critical level in type $A$ was further investigated in \cite{m:cc},
where its explicit generators were produced. It was shown in \cite{m:wa} that
analogues of the affine $W$-algebras can be associated
with the underlying Lie algebras $\agot$
and their description in terms of generators was also given.

\vskip 2mm

Our goal in this paper is to apply a version of the BRST complex $C^k(\lambda,\mu)$
of the quantum Drinfeld--Sokolov reduction to construct a new family of vertex algebras $W^k(\lambda,\mu)$ associated with the centralizers $\agot$
of nilpotent elements
in the Lie algebra $\gl_N$. Here $\la$ is a partition
of $N$ with $n$ parts corresponding to the chosen nilpotent element, while $\mu$ is a partition
of $n$. We show that $W^k(\lambda,\mu)$ inherits a vertex algebra structure from $C^k(\lambda,\mu)$
(Theorem~\ref{thm:vera}). The main results concerning the structure of $W^k(\lambda,\mu)$
are given in Theorem~\ref{thm:degree zero} and Corollary~\ref{cor:affine structure},
where its generating sets were described.

The family of affine $W$-algebras $W^k(\lambda,\mu)$ turns out to interpolate between several
classes of vertex algebras previously studied in the literature.
In the specialization where $\la=(1^N)$ is the column-partition of $N$ and $\mu=(N)$
is the row-partition, we get
the affine $W$-algebra $W^{\tss k}(\gl_N)$ going back to
\cite{z:ie} and \cite{fl:mt}. If the partition $\mu$ of $N$
is arbitrary, then $W^k(\lambda,\mu)$ coincides with the affine
$W$-algebra $W^k(\gl_N,f)$ associated with
a nilpotent element $f \in \gl_N$ of type $\mu$, as introduced in \cite{krw:qr};
see Remark~\ref{rem:krw} below.

In a different specialization where the partition $\lambda$ with $n$ parts is arbitrary
and $\mu=(n)$ is the row-partition,
$W^k(\lambda,(n))$
is the affine $W$-algebra $W^k(\agot)$ introduced in \cite{m:wa}.
For an alternative choice, where
$\mu=(1^n)$ is the column-partition, the $W$-algebra $W^k(\lambda,(1^n))$ coincides
with the universal affine vertex algebra $V^k(\agot)$, where $\agot$ is the centralizer
of a nilpotent element of type $\la$. Note that, in general, the $W$-algebra $W^k(\lambda,\mu)$
need not be conformal; see Example~\ref{rem:no conformal vector}.

\vskip 2mm

Using the same partitions $\lambda$ and $\mu$ as above, we also introduce a family of associative algebras
$U(\lambda,\mu)$ which we call the {\em generalized finite $W$-algebras}; see
Definition~\ref{def:generalized W}. Our main result concerning
these objects is Theorem~\ref{thm:zhu and finite} which connects the affine
$W$-algebra $W^k(\lambda,\mu)$ with $U(\lambda,\mu)$ via the Zhu functor
by analogy with \cite{dk:fa}. In our notation, the construction in {\em loc.~cit.}
corresponds to $\la=(1^N)$ and a nilpotent element $f\in\gl_N$
of type $\mu$; see Example~\ref{ex:finwa} below.
Note that
the finite $W$-algebras $U(\lambda,\mu)$ with $\la=(1^N)$
were also described from a different viewpoint
in \cite{bk:sy} (denoted there by $W(\pi)$ for the {\em pyramid} $\pi$ associated with $\mu$)
by using the {\em shifted Yangians}; the particular case
of rectangular pyramids appeared earlier in \cite{rs:yr}.

As with the affine case, the algebras $U(\lambda,\mu)$ interpolate between several
well-studied objects. The specialization with $\mu=(1^n)$ and arbitrary $\lambda$ recovers
the universal enveloping algebra $U(\agot)$
of the centralizer
$\agot$ of a nilpotent element of type $\la$ (Example~\ref{ex:env}). The specialization $\lambda_1=\lambda_2=\cdots = \lambda_n=p$ with arbitrary $\mu$ recovers the finite $W$-algebras associated with truncated current Lie algebras, also known as generalized Takiff algebras \cite{t:rp}; see Remark~\ref{rem:takiff}. In particular, the Whittaker modules for $U(\lambda, \mu)$ in this setting were studied in \cite{h:tc}. The specialization where $\mu=(n)$ (which we refer to as the `principal
nilpotent case') is considered in Theorem~\ref{thm:principal case},
where we obtain an explicit description
of generators of $U(\lambda,(n))$, then show
that this algebra is isomorphic to the center of $U(\agot)$ (Corollary~\ref{cor:cent}).
The very last Section~\ref{subsec:minnilp} is devoted to the `minimal nilpotent case'
with $\mu=(1^{n-2}2)$ and arbitrary $\la$. We give explicit descriptions of
generators of both the finite and affine $W$-algebras $U(\lambda,\mu)$ and $W^k(\lambda,\mu)$;
see the respective Theorems~\ref{thm:finite generators} and \ref{thm:affine generators}.

\subsection{Further directions}

As written above, in this paper, we introduce new family of vertex algebras which interpolate known vertex algebras such as $V^k(\gl_N)$, $V^k(\mathfrak{a})$, and  $W^k(\gl_N,f).$ Furthermore, we showed some of generalized $W$-algebras are chiralization of some known associative algebras. As an example, a finite $W$-algebra associated with a Takiff algebra is obtained as a Zhu algebra of 
$W^k(\lambda,\mu)$ for a rectangular shape partition $\lambda$.

As a future project, we will observe the relation between the 
$W$-algebras associated with $\mathfrak{gl}_N$ and $\mathfrak{a}.$ 
In particular, when 
$f$ is a nilpotent element of 
 $\gl_N$ inducing the grading \eqref{eq:n_lambda,mu},
the BRST complex for $W^k(\lambda,\mu)$ is contained in the BRST complex for $W^k(\gl_N,f)$. Moreover, the conformal vector on the $W^k(\gl_N,f)$ induces the conformal weight decomposition of  $W^k(\lambda,\mu)$. 
Hence comparing the two vertex algebras should be an interesting problem and also  this observation can provide a quantum affine analogue of the relation between $S(\g)^\g$ and $S(\mathfrak{a})^{\mathfrak{a}}$ established in \cite{ppy:si}. Another interesting possible link might be observed between $W^k(\lambda,\mu)$ and $W^k(\mathfrak{gl}_N,f_{\lambda,\mu})$ where $f_{\lambda,\mu}$ is dual to $\mathsf{e}_{\lambda,\mu}$ in \eqref{eq:e,f_mu}. We expect that $W^k(\lambda,\mu)$ embeds into $W^k(\mathfrak{gl}_N,f_{\lambda,\mu})$.


We also believe that both the vertex algebras $W^k(\lambda,\mu)$ and associative algebras $U(\lambda,\mu)$ deserve further investigation regarding their structure theory and representations, and moreover Feigin--Frenkel-type duality,
classical counterparts of $W^k(\lambda,\mu)$ and critical level phenomena would be interesting. Another interesting question is
a possible relationship of the algebras $U(\lambda,\mu)$ with suitably modified shifted Yangians
of \cite{bk:sy}.


\section{Generalized finite $W$-algebras} \label{sec:finite}

Throughout the paper $N$ denotes a positive integer and $\mathfrak{g}=\mathfrak{gl}_N$
is the general linear Lie algebra over the field $\mathbb{C}$ of complex numbers.
Take a nilpotent element $e\in \mathfrak{g}$ whose Jordan canonical form has the Jordan
blocks of sizes $\lambda_1\leqslant \lambda_2 \leqslant \dots \leqslant \lambda_n.$
 Consider the left-justified {\em pyramid} corresponding to the partition
$\la=(\la_1,\dots,\la_n)$ of $N$. It
consists of $n$ rows of unit boxes with $\la_1$ boxes in the top row,
 $\la_2$ boxes in the second row, etc.
The associated {\em tableau} is obtained by writing
 the numbers $1,2, \dots, N$ into the boxes of the pyramid $\la$
 consecutively by rows from left to right
 starting from the top row.
For example, the tableau associated with the partition $\lambda = (2,3,5)\vdash 10 = N$ is given by
\begin{equation} \label{pyramid_ex}
    \begin{ytableau}
    1 & 2\\
    3 & 4 & 5\\
    6 & 7 & 8 & 9 & 10
\end{ytableau}
\end{equation}

\medskip

\noindent
We let
$\row_{\la}(a)$ and $\col_{\la}(a)$ denote the row and column number of the box containing
the entry $a$.
We will denote by $q_i$ the number of boxes in
the $i$-th column of the pyramid $\la$ for $i=1,2,\dots, l,$ where $l:= \lambda_n$.

Using the tableau associated with the pyramid $\la$, introduce the element $e\in\g$ by
 \begin{equation} \label{eq:e_lambda}
    e=\sum_{\substack{i=1,\dots, N-1 \\[0.2em]
    \row_{\lambda}(i)=\row_{\lambda}(i+1)}} e_{i\, i+1},
\end{equation}
where the $e_{ij}$ denote the standard basis elements of $\mathfrak{g}$.
In the above example, the parameters are
$(q_1,q_2,q_3, q_4, q_5)=(3,3,2,1,1)$
and the corresponding nilpotent element is
\[e=e_{12}+e_{34}+e_{45}+e_{67}+e_{78}+e_{89}+e_{9\, 10}.
\non
\]

Given a pyramid $\la$ with $n$ rows, take an arbitrary partition $\mu$ of $n$.
As with $\la$, we will write the parts of $\mu = (\mu_1,\mu_2, \dots, \mu_m)$
in the weakly increasing order $\mu_1\leqslant \mu_2\leqslant \dots \leqslant \mu_m$
and consider the corresponding pyramid $\mu$ along with the associated tableau.
For the pyramid $\lambda$ appearing in \eqref{pyramid_ex}, there are three possible
pyramids $\mu$ and associated tableaux:
\beql{muexa}
    \begin{ytableau}
    1 \\
    2 \\
    3
\end{ytableau}
\qquad\qquad
    \begin{ytableau}
    1 \\
    2 & 3\\
\end{ytableau}
\qquad\qquad
\begin{ytableau}
    1 & 2& 3\\
\end{ytableau}
\eeq

The centralizer $\mathfrak{a}:=\mathfrak{g}^e$ of the element $e\in\mathfrak{g}$
defined in \eqref{eq:e_lambda}
is a Lie subalgebra with the basis
\begin{equation} \label{eq:basis_a}
   \mathcal{B}^e:= \{ \ E_{ij}^{(r)}\ | \  (i,j,r)\in S^e \, \},
\end{equation}
where
\ben
    S^e:=\{\, (i,j,r) \in \mathbb{Z}^3\, |\, 1\leqslant i,j\leqslant n\, ,\,
    \ \lambda_j - \text{min}(\lambda_i, \lambda_j)\leqslant r < \lambda_j\, \}
\een
and
\begin{equation}\label{eq:E}
    E_{ij}^{(r)}= \sum_{\substack{\row_{\la}(a)=i,\, \row_{\la}(b)=j\\[0.2em]
    \col_{\la}(b)-\col_{\la}(a)=r} } e_{ab}.
\end{equation}
The commutation relations in the centralizer $\agot$ are given by
\ben
    [E_{ij}^{(r)}, E_{hl}^{(s)}]= \delta_{hj}E_{il}^{(r+s)}-\delta_{il}E_{hj}^{(r+s)},
\een
where we set $E_{ij}^{(r)}=0$ for $(i,j,r)\not \in S^e.$

Now we use the chosen pyramid $\mu$ to introduce
a $\mathbb{Z}$-gradation $\mathfrak{a}:= \bigoplus_{i\in \mathbb{Z}}\mathfrak{a}(i)$
on the centralizer by setting
 \begin{equation} \label{eq:deg_mu}
     \de_{\mu}(E_{ij}^{(r)}):= \col_{\mu}(j)- \col_{\mu}(i).
 \end{equation}
Set
 \begin{equation}\label{eq:n_lambda,mu}
     \mathfrak{n}_{\lambda,\mu}\quad := \quad  \bigoplus_{i>0} \mathfrak{a}(i)
     \quad= \quad \text{Span}_\mathbb{C} \ \{ \ E_{ij}^{(r)} \ | \ (i, j, r) \in S_{\lambda, \mu}\},
 \end{equation}
 where
 \begin{equation} \label{eq:index S}
     S_{\lambda,\mu}:=\{(i,j,r)\in S^e\,  |\,   \col_\mu(i)<\col_\mu(j)\}.
 \end{equation}

Introduce the element
\begin{equation}  \label{eq:e,f_mu}
\mathsf{e}_{\lambda,\mu}:= \sum_{\substack{i=1,\dots, n-1 \\[0.2em] \row_{\mu}(i)
=\row_{\mu}(i+1)}} E_{i\, i+1}^{(\lambda_{i+1}-1)} \in \mathfrak{a}(1)
\end{equation}
associated with the pair of pyramids $(\lambda,\mu)$. To illustrate, note that
the elements associated with the respective tableaux in \eqref{muexa} are
\ben
\mathsf{e}_{\lambda,\mu}=0,\qquad \mathsf{e}_{\lambda,\mu}=E_{23}^{(4)}\Fand \mathsf{e}_{\lambda,\mu}=E_{12}^{(2)}+E_{23}^{(4)}.
\een

Obviously, in the general case the relation $\row_{\mu}(i)
=\row_{\mu}(i+1)$ is possible if and only if $\col_{\mu}(i+1)=\col_{\mu}(i)+1$.
Define the associated character $\chi \in \mathfrak{n}_{\lambda, \mu}^\ast$ by
setting
\beql{defchi}
\chi\big(E_{i\, i+1}^{(\lambda_{i+1}-1)}\big)=1
\eeq
\noindent
if $\row_{\mu}(i)=\row_{\mu}(i+1)$ with $i\in\{1,\dots,n-1\}$, and $\chi(E_{ij}^{(r)})=0$
for all remaining triples $(i,j,r)\in S_{\lambda,\mu}$. For an explanation of this choice of $\chi$, see Remark~\ref{rem:chi}.
Consider the left ideal
\ben
    \mathcal{I}_{\lambda,\mu}:= U(\mathfrak{a})\left< \left.
 \mathsf{n}+\chi(\mathsf{n})\right| \mathsf{n} \in \mathfrak{n}_{\lambda,\mu}\, \right>
\een
of the universal enveloping algebra $U(\mathfrak{a})$. Because $\chi\in \mathfrak{n}_{\lambda,\mu}^*$ with $\chi|_{\mathfrak{a}(i)}=0$ for all $i \geqslant 2$, the adjoint action of
$\mathfrak{n}_{\lambda,\mu}$ on $U(\mathfrak{a})$ induces the action of $\mathfrak{n}_{\lambda,\mu}$
on the quotient $U(\mathfrak{a})/\mathcal{I}_{\lambda,\mu}$. As with the finite $W$-algebras
(cf. \cite{bk:sy} and \cite[Appendix]{dk:fa}), the subspace of $\mathfrak{n}_{\lambda,\mu}$-invariants
turns out to be an associative algebra with the product inherited from $U(\agot)$ as we verify below.

\begin{defn}  \label{def:generalized W}
The generalized finite $W$-algebra is the associative algebra
\[ U(\lambda,\mu)= (U(\mathfrak{a})/\mathcal{I}_{\lambda,\mu})^{\mathfrak{n}_{\lambda,\mu}}.\]
\end{defn}

To verify that $U(\lambda, \mu)$ is a well-defined associative algebra,
take $A, B \in U(\mathfrak{a})$ such that
    $\overline{A}, \overline{B} \in U(\lambda, \mu)$ are their
    images in the quotient, and $\mathsf{n} \in \mathfrak{n}_{\lambda, \mu}$.
    It is sufficient
    to show that $[\mathsf{n}, AB] \in \mathcal{I}_{\lambda, \mu}$.
    Since $\overline{A} \in U(\lambda, \mu)$, we have
    $[\mathsf{n}, A] = a(\mathsf{m} + \chi(\mathsf{m}))$
    for some $a \in U(\mathfrak{a})$ and $\mathsf{m} \in \mathfrak{n}_{\lambda, \mu}$.
    By the Leibniz rule,
    \ben
        \begin{aligned}
            \relax [\mathsf{n}, AB] &= [\mathsf{n}, A]B + A[\mathsf{n}, B] \\
            &= a(\mathsf{m} + \chi(\mathsf{m}))B + A[\mathsf{n}, B] \\
            &= aB(\mathsf{m} + \chi(\mathsf{m})) + a[\mathsf{m}, B]
            + A[\mathsf{n}, B] \in \mathcal{I}_{\lambda, \mu},
        \end{aligned}
    \een
    as required.

\begin{ex}\label{ex:finwa}
    Let $\lambda=(1,1,\dots, 1)$ be the column-pyramid with $N$
    boxes so that $\mathfrak{a} = \mathfrak{gl}_N$ and let
    $\mu$ be a pyramid with $N$ boxes.
    Consider the element $\mathsf{e}_{\lambda,\mu}$ associated with $\mu$ as defined in \eqref{eq:e,f_mu}.
     Then
    \[U(\lambda,\mu)=U(\mathfrak{g}, \mathsf{e}_{\lambda,\mu})\]
    is the finite $W$-algebra
    associated with $\mathfrak{g}$ and $\mathsf{e}_{\lambda,\mu}$; see \cite{gg:qs}, \cite{k:wv} and \cite{p:st}.
\end{ex}

\begin{ex}\label{ex:env}
    Let $\lambda$ be a pyramid consisting of $N$ boxes with $n$
    rows and $\mu=(1,1,\dots, 1)$ be the column-pyramid with $n$ boxes.
     Then
    \[  U(\lambda,\mu)=U(\mathfrak{a}),\]
    where $\mathfrak{a}$ is the centralizer of the nilpotent element $e$ defined in
    \eqref{eq:e_lambda}.
\end{ex}

\begin{rem} \label{rem:takiff}
    Let $N=np$ and $\lambda_1=\lambda_2= \cdots=\lambda_n=p$ and $\mu$ be a pyramid with $n$ boxes. We have a Lie algebra isomorphism
    \begin{equation}\label{eq:iso}
        \mathfrak{a} \ \cong \ \mathfrak{gl}_n[v]/(v^p), \qquad E_{ij}^{(r)} \longleftrightarrow e_{ij} v^r
    \end{equation}
    Here, $\mathfrak{gl}_n[v]/(v^p)$ denotes the truncated current Lie algebra of level $p$, equipped with the Lie bracket induced from the current Lie algebra $\mathfrak{gl}_n[v]$. The $W$-algebras associated with $\mathfrak{gl}_n[v]/(v^p)$ and a nilpotent element $e$ were constructed in \cite{h:tc}. The $W$-algebra in \cite{h:tc} is identical to the $W$-algebra constructed in this paper from the Lie algebra $\mathfrak{a}$, after identifying via \eqref{eq:iso}. (See \cite{cw:rt} and references therein.)  
\end{rem}

\section{Generalized affine $W$-algebras}
\label{sec:gaff}

\subsection{Affine vertex algebra and free fermion vertex algebra}
\label{subsec:ava}

  Let $e$ be the nilpotent element in $\mathfrak{g}$ associated with
  the pyramid $\lambda$ as defined in \eqref{eq:e_lambda}. As before, we denote
  by $\mathfrak{a}$ be
  the centralizer of $e$ in $\mathfrak{g}.$
  We will use the invariant bilinear form $( \ | \  )$ on
  $\mathfrak{a}$ which was introduced in \cite{ap:qm}.
  Specifically, set $\de_{\lambda}(e_{ij}) = \col_\lambda(j) - \col_\lambda(i)$
  to induce the $\mathbb{Z}$-gradations
\ben
\mathfrak{g} = \bigoplus_{r \in \mathbb{Z}} \mathfrak{g}_r\qquad\text{and}\qquad
\mathfrak{a} = \bigoplus_{r \in \mathbb{Z}} \mathfrak{a}_r
  = \bigoplus_{r \in \mathbb{Z}} (\mathfrak{a} \cap \mathfrak{g}_r).
\een
  Then for homogeneous elements $X,Y\in\agot$ we set
\begin{equation} \label{eq:bilinearform}
    ( X | Y ) =\begin{cases} \frac{1}{2N}
    \text{tr}_{\mathfrak{g}_0} \big((\text{ad }X) (\text{ad } Y) \big)
    \quad &\text{if } X, Y \in \mathfrak{a}_0, \\ 0 \quad &\text{otherwise.}\end{cases}
\end{equation}

\begin{rem} \label{rem:bilinear}
    We can consider another bilinear form induced from that of $\mathfrak{gl}_N= \mathfrak{sl}_N \oplus \mathbb{C}.$ The main results in this paper such as  Theorem \ref{thm:degree zero}, Corollary \ref{cor:affine structure}, and Theorem \ref{thm:zhu and finite} are still true even when this bilinear form is used. However, we follow the convention used in \cite{ap:qm}, \cite{mr:cw}, and \cite{m:wa}.
\end{rem}

To use a version of the BRST complex
of the quantum Drinfeld--Sokolov reduction, associated with the Lie algebra $\agot$
by analogy with \cite[Ch.~15]{fb:va}, and extending the construction of \cite{m:wa},
introduce two vertex algebras
as follows. The first is $V^k(\mathfrak{a})$, the {\em affine vertex algebra at the level
$k\in \mathbb{C}$}. It is defined as
    \ben
    V^k(\mathfrak{a}) = U(\widehat{\mathfrak{a}})
    \otimes_{U(\mathfrak{a}[t])\oplus \mathbb{C} K} \mathbb{C}_k \, ,
    \een
    where $\widehat{\mathfrak{a}} = \mathfrak{a}[t, t^{-1}]
    \otimes \mathbb{C}K$ is the affine Kac--Moody algebra with
    central element $K$, whose commutation relations are given by
    \begin{equation} \label{eq:affine_comm}
        [a \otimes t^n, b \otimes t^m] = [a, b] \otimes t^{n+m} + \delta_{n+m, 0} (a|b) K
    \end{equation}
    using the bilinear form \eqref{eq:bilinearform}. Here $\mathbb{C}_k$
    is the one-dimensional module of $U(\mathfrak{a}[t]) \oplus \mathbb{C}K$,
    where all elements act trivially except $K$, which acts as
    $k \in \mathbb{C}$. For $(E_{ij}^{(r)} \otimes t^{-1})\in V^k(\mathfrak{a})$,
    the corresponding field is
    $E_{ij}^{(r)}(z) := \sum_{n \in \mathbb{Z}} (E_{ij}^{(r)} \otimes t^n) z^{-n-1}$
    and the commutation relations induced from \eqref{eq:affine_comm} are
    \ben
        [E_{i_1, j_1}^{(r_1)}(z), E_{i_2, j_2}^{(r_2)}(w)]
        = [E_{i_1, j_1}^{(r_1)}, E_{i_2, j_2}^{(r_2)}](w)
        \delta(z,w) + k(E_{i_1, j_1}^{(r_1)} \ | \ E_{i_2, j_2}^{(r_2)}) \partial_w\delta(z,w),
    \een
    where $\delta(z,w)=\sum_{m\in \mathbb{Z}}z^m w^{-m-1}.$
    In terms of the $\lambda$-bracket, they can be rewritten as
    \ben
        [E_{i_1, j_1}^{(r_1)} {}_\lambda E_{i_2, j_2}^{(r_2)}]
        = [E_{i_1, j_1}^{(r_1)}, E_{i_2, j_2}^{(r_2)}]
        + k(E_{i_1, j_1}^{(r_1)} \ | \ E_{i_2, j_2}^{(r_2)}) \ts\lambda;
    \een
    see e.g. \cite{k:iv}.

The second vertex algebra is
$\mathcal{F}(\mathfrak{n}_{\lambda,\mu})$, the
{\em free fermion vertex algebra associated with $\mathfrak{n}_{\lambda, \mu}$}.
    More precisely, let us consider the dual space
    $\mathfrak{n}_{\lambda, \mu}^\ast$ of $\mathfrak{n}_{\lambda, \mu}$ and the
    odd vector superspaces
\begin{equation} \label{eq:phi}
         \phi_{\mathfrak{n}_{\lambda, \mu}}=\{ \phi_\mathsf{n}\, | \,
         \mathsf{n}\in \mathfrak{n}_{\lambda, \mu}  \} \Fand
          \phi^{\mathfrak{n}_{\lambda, \mu}^\ast}=\{ \phi^\mathsf{m} \, | \,
          \mathsf{m}\in \mathfrak{n}_{\lambda, \mu}^\ast  \}.
\end{equation}
We set for brevity,
\ben
    \phi_{(i,j,r)}:= \phi_{E_{ij}^{(r)}}\Fand \phi^{(i,j,r)}:= \phi^{E_{ij}^{(r)\ast}},
\een
where the $E_{ij}^{(r)\ast}$ are the elements of the basis of $\mathfrak{n}_{\lambda, \mu}^\ast$
dual to the basis formed by the elements
$E_{ij}^{(r)} \in \mathfrak{n}_{\lambda, \mu}$.
The vertex algebra $\mathcal{F}(\mathfrak{n}_{\lambda,\mu})$ is freely generated
by  the elements $\phi_{(i, j, r)}$ and $\phi^{(i, j, r)}$ for $(i, j, r) \in S_{\lambda, \mu}$
as a differential algebra with the following $\lambda$-brackets:
\begin{equation} \label{eq:ff OPE}
\begin{aligned}[]
    [\phi_{(i, j, r)}{}_\lambda \phi^{(i', j', r')}]&
    =[\phi^{(i', j', r')}{}_\lambda \phi_{(i, j, r)}]=\delta_{ii'}\delta_{jj'}\delta_{rr'}, \\
    [\phi^{(i, j, r)}{}_\lambda \phi^{(i', j', r')}]&
    = [\phi_{(i, j, r)}{}_\lambda \phi_{(i', j', r')}]=0.
\end{aligned}
\end{equation}
The first relation in \eqref{eq:ff OPE} can also be written as
\ben
    [\phi_{\mathsf{n}}{}_\lambda \phi^{\mathsf{m}}]
    = [\phi^{\mathsf{m}} {}_\lambda \phi_{\mathsf{n}}]= {\mathsf{m}}({\mathsf{n}})
\een
for ${\mathsf{n}}\in \mathfrak{n}_{\lambda,\mu}$ and
${\mathsf{m}}\in \mathfrak{n}_{\lambda,\mu}^\ast.$
Also, we set $\phi_x:= \phi_{\pi_+(x)}$ where
$\pi_+ : \mathfrak{a} \to \mathfrak{n}_{\lambda, \mu}$
is the projection to $\mathfrak{n}_{\lambda,\mu}$
which is zero on the basis vectors which do not belong to $\mathfrak{n}_{\lambda,\mu}$.

\subsection{Definition of the generalized affine $W$-algebras}

\label{subsec:dda}

In this section, we introduce the generalized affine $W$-algebras via BRST cohomologies.
Consider the vertex algebra
\ben
    C^k(\lambda, \mu):=  V^k(\mathfrak{a}) \otimes \mathcal{F}(\mathfrak{n}_{\lambda,\mu})
\een
and its element
\ben
    d = \sum_{I\in S_{\lambda,\mu}} :\phi^{I} E_I: + \phi^{\chi}
    + \frac{1}{2}\sum_{I,I' \in S_{\lambda,\mu}}:\phi^{I}\phi^{I'}\phi_{[E_{I'},E_{I}]}:
    \ \in \ C^k(\lambda,\mu).
\een
Here $E_I:=E_{ij}^{(r)}$ and $\phi^I:=\phi^{(i,j,r)}$ for $I=(i,j,r)\in S_{\lambda, \mu}$;
see also \eqref{defchi} for the definition of $\chi$.
Let us denote
\ben
    Q:= d_{(0)} \ \in \ \text{End}\ts C^k(\lambda, \mu),
\een
where $d_{(0)}: A \mapsto [d\, {}_\lambda \,  A]\, \big|_{\lambda=0}.$
By the fundamental property of $\lambda$-bracket,  $Q$ is an odd derivation
with respect to the normally ordered product and commutes with $\partial$.
Hence the following lemma completely determines $Q.$

\begin{lem}\label{lem:coa}
The following formulas hold:
\begin{enumerate}[(1)]
    \item $Q(a) = \sum_{I\in S_{\lambda, \mu}} :\phi^I [E_I, a]:
    + \sum_{(i, j, r) \in S_{\lambda, \mu}} k (a | E_I)\partial \phi^I$ for $a\in \mathfrak{a}$.
    \item $Q( \phi_{\mathsf{n}}) = \mathsf{n} + \chi(\mathsf{n})
    + \sum_{I \in S_{\lambda, \mu}} : \phi^{I}
    \phi_{[E_I, \mathsf{n}]}:$ for ${\mathsf{n}}\in \mathfrak{n}_{\lambda, \mu}$.
    \item $Q (\phi^{\mathsf{m}}) = \frac{1}{2}
    \sum_{I \in S_{\lambda,\mu}} :\phi^{I} \phi^{E_{I} \cdot \mathsf{m}} : $
    for ${\mathsf{m}}\in \mathfrak{n}_{\lambda, \mu}^\ast$.
\end{enumerate}
In the last equation, $E_I \cdot \mathsf{m}$ denotes the coadjoint
action of $\mathfrak{n}_{\lambda, \mu}$ on $\mathfrak{n}_{\lambda,\mu}^\ast$.
\end{lem}
\begin{proof}
It is enough to compute $\lambda$-brackets between $d$ and generators of $C^k(\lambda,\mu).$
    \begin{enumerate}
        \item[(1)] Observe that $[d{}_{\lambda} a]=
        \sum_{I\in S_{\lambda, \mu}} [ \,  :\phi^{I} E_I:\, {}_\lambda \, a \, ].$
        By skew-symmetry, we have
        \begin{equation}
        \non
        \begin{aligned}[]
             [d{}_{\lambda} a]  &= - \sum_{I\in S_{\lambda, \mu}}
             \big[a_{-\lambda-\partial} :\phi^{I} E_I:\big]
             = \sum_{I\in S_{\lambda, \mu}} :\phi^{I}[E_I, a]:
             + \sum_{I \in S_{\lambda, \mu}} k(\partial+ \lambda) (a | E_I)\phi^{I}.
        \end{aligned}
        \end{equation}
        \item[(2)] By direct computations, we get
        $  \big[ \phi^{\chi}\, {}_\lambda \ \phi_{\mathsf{n}}\big] = \chi(\mathsf{n})$ and
        \begin{equation}
        \non
            \begin{aligned}
                \sum_{I\in S_{\lambda, \mu}}
                \big[ : \phi^{I} E_I: {}_\lambda \ \phi_{\mathsf{n}}\big]
                &= \sum_{I\in S_{\lambda, \mu}} : [\phi_{\mathsf{n} \ -\lambda-\partial }
                \ \phi^{I}] \, E_I:  = \mathsf{n}.
            \end{aligned}
        \end{equation}
    We also have
        \begin{equation}
        \non
            \begin{aligned}
                & \sum_{{I,I'} \, \in S_{\lambda,\mu} }\big[:\phi^{I}
                \phi^{I'}\phi_{[E_{I'},E_{I}]}:\, {}_\lambda \phi_{\mathsf{n}}\big]\\
                &=  \sum_{I,I' \in S_{\lambda,\mu}} :[\phi_{\mathsf{n}
                \ -\lambda-\partial} \phi^{I}]: \phi^{I'}
                \phi_{[E_{I'}, E_{I}]}:: - \sum_{I,I' \,
                \in S_{\lambda,\mu}}:\phi^{I} [\phi_{\mathsf{n}
                \ -\lambda-\partial} :\phi^{I'}\phi_{[E_{I'}, E_{I}]}:]:\\
                &=  \sum_{I' \in S_{\lambda, \mu}} :\phi^{I'}
                \phi_{[E_{I'}, \mathsf{n}]}: -
                \sum_{I \in S_{\lambda, \mu}} :\phi^{I} \phi_{[\mathsf{n}, E_{I}]}: \
                = 2\sum_{(i, j, r) \in S_{\lambda, \mu}} :\phi^{I} \phi_{[E_{I}, \mathsf{n}]}:  \ .
            \end{aligned}
        \end{equation}
        Hence we conclude that
        \ben
        [d_{\lambda} \phi_{\mathsf{n}}] = \mathsf{n} +\chi(\mathsf{n})
        + \sum_{I \in S_{\lambda, \mu}} : \phi^{I} \phi_{[E_{I}, \mathsf{n}]}: \ .
        \een
        \item[(3)] We have
        \begin{equation}
        \non
            \begin{aligned}[]
                [d{}_{\lambda}\phi^{\mathsf{m}}] &= \frac{1}{2}
                \sum_{I,I'  \in S_{\lambda,\mu}} :\phi^{I}:
                \phi^{I'} \big[ \phi^{\mathsf{m}}_{\ \lambda} \phi_{[E_{I'}, E_{I}]}\big]::\\
                &= \frac{1}{2} \sum_{I,I' \in S_{\lambda,\mu}}
                \mathsf{m}\left( [E_{I'}, E_{I}]\right) : \phi^{I} \phi^{I'} : \
                = \frac{1}{2} \sum_{I \in S_{\lambda,\mu}} :\phi^{I} \phi^{E_{I} \cdot \mathsf{m}} : \ .
            \end{aligned}
        \end{equation}
    \end{enumerate}
\end{proof}

\begin{prop} \label{prop:differential}
    The endomorphism $Q$ is a differential on $ C^k(\lambda, \mu)$; that is, $Q^2=0.$
\end{prop}

\begin{proof}
   To verify the relation $Q^2=0$, it is enough to show that $[d{}_\lambda d]=0$.
   We have
   \ben
    \bal
        & \sum_{I \in S_{\lambda, \mu}} \big[\, d\, {}_\lambda :\phi^{I} E_I:\big]
        = \sum_{I \in S_{\lambda, \mu}} \Big(:[d_\lambda \phi^{I}] E_I:
        -  :\phi^{I} [d_\lambda E_I]: + \int_{0}^\lambda [[d_\lambda \phi^{I}]_\mu E_I] \ d\mu\Big) \\
        &= \sum_{I,I' \in S_{\lambda,\mu}} \frac{1}{2} ::\phi^{I'} \phi^{E_{I'} \cdot E_{I}^{\ast}}: E_I:
        - \sum_{I,I' \in S_{\lambda,\mu}} :\phi^{I} :\phi^{I'}[E_{I'}, E_{I}] :: \\
        &= \sum_{I, I' \in S_{\lambda, \mu}} \frac{1}{2} ::\phi^{I} \phi^{I'}: [E_{I'}, E_I]:
        - \sum_{I, I' \in S_{\lambda, \mu}} :\phi^{I} :\phi^{I'} [E_{I'}, E_{I}]:: \\
        &= -\frac{1}{2}\sum_{I,I'\in S_{\lambda,\mu}} :\phi^{I}:\phi^{I'} [E_{I'}, E_I]:: \ .
    \eal
    \een
    One can easily check that
$[d_\lambda \phi^{\chi}]=0$ and hence it remains to verify that
\begin{equation} \label{eq:d^2=0-2}
    \frac{1}{2}\sum_{I,I' \in S_{\lambda,\mu}}
    \big[d_\lambda :\phi^{I} \phi^{I'} \phi_{[E_{I'}, E_{I}]}:\big]
    =-\sum_{I \in S_{\lambda, \mu}}\big[d_\lambda  :E_I \phi^{I}:\big].
\end{equation}
Expand the left hand side of \eqref{eq:d^2=0-2} as follows:
        \begin{align}
            &\frac{1}{2}\sum_{I,I' \in S_{\lambda,\mu}}\big[d_\lambda :\phi^{I}
            \phi^{I'} \phi_{[E_{I'}, E_{I}]}:\big]  \label{eq:d^2=0-3-1} \\
            &= \frac{1}{4} \sum_{I,I',I'',I''' \in S_{\lambda,\mu}} (E_I)^\ast
            \left( [E_{I'''}, E_{I''}] \right) ::\phi^{I''} \phi^{I'''}:
            \ :\phi^{I'}\phi_{[E_{I'}, E_{I}]}:: \label{eq:d^2=0-3-2}\\
            &- \frac{1}{4} \sum_{I,I',I'',I''' \in S_{\lambda,\mu}} (E_{I'})^\ast
            \left( [E_{I'''}, E_{I''}]\right) : \phi^{I}
            (:: \phi^{I''} \phi^{I'''} : \phi_{[E_{I'}, E_{I}]}:): \label{eq:d^2=0-3-3}\\
            & + \frac{1}{2}\sum_{I,I' \in S_{\lambda,\mu}} :\phi^{I} :\phi^{I'} [E_{I'}, E_{I}]::
            \label{eq:d^2=0-3-4}\\
            &+ \frac{1}{2} \sum_{I,I'\in S_{\lambda,\mu}} \chi
            \left( [E_{I'}, E_{I}] \right) : \phi^{I} \phi^{I'}: \label{eq:d^2=0-3-5}\\
            & + \frac{1}{2} \sum_{I,I',I''\in S_{\lambda,\mu}} :\phi^{I} :\phi^{I'} :\phi^{I''}
            \phi_{[E_{I''}, [E_{I'}, E_{I}]]}::: \label{eq:d^2=0-3-6}\\
            &- \frac{1}{2} \sum_{I,I',I'' \in S_{\lambda,\mu}} E_{I'}^\ast
            ([[E_{I'}, E_{I}], E_{I''}]) :\phi^{I} \phi^{I''}: \ \lambda\label{eq:d^2=0-3-7}\\
            &+ \frac{1}{2} \sum_{I,I',I'' \in S_{\lambda,\mu}} E_{I'}^\ast
            ([[E_{I''}, E_{I'}], E_I]) : \phi^{I} \phi^{I''}: \  \lambda \ . \label{eq:d^2=0-3-8}
        \end{align}
    The expression in \eqref{eq:d^2=0-3-5} is zero because
    $ \text{deg}([E_{I'}, E_{I}]) \geqslant  2$
    (if the commutator is nonzero). The sum of the terms
    in \eqref{eq:d^2=0-3-7} and \eqref{eq:d^2=0-3-8} is zero, as follows by relabeling
    the indices $I$ and $I''$ in \eqref{eq:d^2=0-3-8}, while
    the sum of the terms in
    \eqref{eq:d^2=0-3-2}, \eqref{eq:d^2=0-3-3} and \eqref{eq:d^2=0-3-6} is zero
    because of the Jacobi identity. Hence \eqref{eq:d^2=0-2} follows.
\end{proof}

By Proposition \ref{prop:differential}, the cohomology $H(C^k(\lambda,\mu), Q)$ is well-defined.
This enables us to state the key definition.

\begin{defn}\label{def:affwa}
The vertex algebra
    \[W^k(\lambda,\mu)= H(C^k(\lambda,\mu), Q)\]
    is called the generalized affine $W$-algebra of level $k$ associated
    with the partitions $\lambda$ and $\mu$.
\end{defn}

Recall that the operator $Q=d_{(0)}$ commutes with the endomorphism $\partial$.
Note also that $Q$
is a derivation with respect to the normally ordered product and
$\lambda$-bracket. In other words, we have the properties
\begin{align}
    & Q(:AB:)=:Q(A)B:+(-1)^{p(A)}:A \, Q(B): \label{eq:derv_norm}\\
    & Q([A\, {}_\lambda \, B])= [Q(A)_\lambda B] +(-1)^{p(A)}[A_\lambda Q(B)]
    \label{eq:derv_lambda}
\end{align}
for $A,B\in C^k(\lambda,\mu).$ Indeed, \eqref{eq:derv_norm} follows from
the Wick formula, whereas \eqref{eq:derv_lambda} follows from the Jacobi identity
of Lie conformal algebras. These two properties
imply that if $A$ and $B$ in $C^k(\lambda,\mu)$ are representatives of
elements in $W^k(\lambda,\mu)$, i.e. $Q(A)=Q(B)=0$, then $Q(:AB:)=Q([A{}_{\lambda}B])=0$.
Furthermore, we have the following lemma which implies
that a vertex algebra structure on $W^k(\lambda,\mu)$ is well-defined.

\begin{lem}\label{prop:BRST VA structure}
    \begin{enumerate}[(1)]
        \item[]
        \item If $A \in C^k(\lambda, \mu)$ is in the image of $Q$ then $\partial A$ is also in the image of $Q$.
        \item If $B \in \ker Q$, then $:Q(A)B:$ is in the image of $Q$ for $A \in C^k(\lambda, \mu)$.
        \item If $B \in \ker Q$ and $A \in C^k(\lambda, \mu)$ then $[Q(A)_\lambda B]$
        is in $\mathbb{C}[\lambda] \otimes \text{im}(Q)$.
    \end{enumerate}
\end{lem}

\begin{proof}
 Suppose $Q(B)=A$ for some $B\in C^k(\lambda,\mu).$ Since $Q(\partial B)=\partial A,$
 the first assertion follows. The second and third assertions follow
 from \eqref{eq:derv_norm} and \eqref{eq:derv_lambda}, respectively, by
 taking into account the relation $[A_\lambda Q(B)]=0.$
\end{proof}

\begin{thm}
\label{thm:vera}
    The vertex algebra structure of $C^k(\lambda,\mu)$ induces the vertex
    algebra structure on $W^k(\lambda, \mu)$.
\end{thm}

\begin{proof}
    By Lemma \ref{prop:BRST VA structure}, $W^k(\lambda, \mu)$ is closed
    under the derivation $\partial$, normally ordered product, and the $\lambda$-bracket.
\end{proof}

\subsection{Structure of the generalized affine $W$-algebra}

\label{subsec:sga}

In this section, we describe some properties of a minimal generating set
of the vertex algebra $W^k(\lambda,\mu)$ by introducing another
construction via a subcomplex of $C^k(\lambda,\mu).$
Introduce the `building blocks'
\begin{equation} \label{eq:building block}
    J_a = a + \sum_{I \in S_{\lambda, \mu}} :\phi^{I} \phi_{[E_{I}, a]}:
\end{equation}
for $a \in \mathfrak{a}$ and let $\pi_\leqslant a := a - \pi_+(a)$.
By direct computations, we can check that
    \begin{equation}  \label{eq:d with building block}
    \begin{aligned}
        \relax  [d_\lambda J_a] &=
        \sum_{I \in S_{\lambda, \mu}} :\phi^{I}J_{\pi_\le[E_{I}, a]}:
        - \ \phi^{a \cdot \chi} \\&+(\lambda + \partial)
        \sum_{I \in S_{\lambda, \mu}}\Big( k(E_{I} | a)
        + \text{tr}_{\mathfrak{n}_{\lambda, \mu}} (\pi_+  \text{ad } E_{I}
        \circ \pi_+  \text{ad } a) \Big)\phi^{I}
        \end{aligned}
    \end{equation}
    where $a \cdot \chi$ is the coadjoint action of $\mathfrak{n}_{\lambda, \mu}$
    on $\mathfrak{n}_{\lambda, \mu}^\ast$.
    Also, we have
    \begin{equation} \label{eq: phi with building block}
        [J_a{}_\lambda \phi_{\mathsf{n}}]=\phi_{[a, \mathsf{n}]} ,
        \quad  [J_a{}_\lambda \phi^{\mathsf{m}}]= \phi^{a \cdot \mathsf{m}} \ ,
    \end{equation}
and
    \ben
        \begin{aligned}
            \relax [J_{a \ \lambda} J_b] &=J_{[a, b]} + (k(a|b)
            +\text{tr}_{\mathfrak{n}_{\lambda, \mu}}(\pi_+ \text{ad }a \circ  \pi_+ \text{ad }b)) \lambda \\
            &- \sum_{I \in S_{\lambda, \mu}} :\phi^{I}
            \phi_{[\pi_\le[E_{I}, a], b]}:
            + \sum_{I \in S_{\lambda, \mu}} :\phi^{I} \phi_{[\pi_\le[E_{I}, b], a]}: \ .
        \end{aligned}
    \een
In particular, if $a$ and $b$ are both in
$\mathfrak{n}_{\lambda,\mu}$ or in the subspace
\beql{subp}
\mathfrak{p}:=\bigoplus_{i\leqslant 0} \mathfrak{a}(i),
\eeq
then we have
\begin{equation}\label{eq: two building block}
 [J_a {}_\lambda J_b]= J_{[a,b]}+ (k(a|b)
 + \text{tr}_{\mathfrak{n}_{\lambda, \mu}} (\pi_+  \text{ad } a \circ \pi_+  \text{ad } b))\lambda.
\end{equation}

We have the tensor product decomposition of
the vertex algebra $C^k(\lambda,\mu)$ given by
\ben
    C^k(\lambda,\mu)= C_+^k(\lambda,\mu)\otimes \widetilde{C}^k(\lambda,\mu),
\een
where $C_+^k(\lambda,\mu)$ and $\widetilde{C}^k(\lambda,\mu)$
are the vertex subalgebras respectively generated by the subspaces
$r_+:=\phi_{\mathfrak{n}_{\lambda, \mu}}\oplus J_{\mathfrak{n}_{\lambda, \mu}}$ and
$r_-:=\phi^{\mathfrak{n}_{\lambda,\mu}^\ast}\oplus J_{\mathfrak{p}}$.
By \eqref{eq: phi with building block} and \eqref{eq: two building block},
we can conclude $C_+^k(\lambda,\mu)$ and $\widetilde{C}^k(\lambda,\mu)$
are universal enveloping vertex algebras of the nonlinear Lie conformal
algebras
\ben
R_+:= \mathbb{C}[\partial]\otimes r_+\Fand
R_-:= \mathbb{C}[\partial]\otimes r_-.
\een
Moreover, the differential $Q$ has the properties
\ben
Q(C_+^k(\lambda,\mu)) \subset C_+^k(\lambda,\mu)\Fand
Q(\widetilde{C}^k(\lambda,\mu)) \subset\widetilde{C}^k(\lambda,\mu).
\een
Due to the fact that $H(C_+^k(\lambda,\mu), Q|_{C_+^k(\lambda,\mu)})= \mathbb{C}$, we then have
\ben
 H(C^k(\lambda,\mu), Q) \simeq H(C_+^k(\lambda,\mu), Q|_{C_+^k(\lambda,\mu)})
 \otimes  H(\widetilde{C}^k(\lambda,\mu), \widetilde{Q})
 \simeq  H(\widetilde{C}^k(\lambda,\mu), \widetilde{Q}),
\een
where $\widetilde{Q}:=Q|_{\widetilde{C}^k(\lambda,\mu)}$.
As a conclusion, the following proposition holds.

\begin{prop}\label{prop:isom}
The generalized $W$-algebra $W^k(\lambda, \mu)$
is isomorphic to $ H(\widetilde{C}^k(\lambda,\mu), \widetilde{Q}).$
\end{prop}

Define the conformal weight $\Delta$
on the complex $\widetilde{C}^k(\lambda, \mu)$ induced by
\begin{equation} \label{eq:Delta in W}
    \Delta(J_a)= 1-j_a, \quad \Delta(\phi^{I})=\text{deg}(E_{I}), \quad \Delta(\partial)=1
\end{equation}
for $a\in \mathfrak{a}(j_a)$. We denote the subspace of conformal
weight $\Delta$ of $\widetilde{C}^k(\lambda,\mu)$ by $\widetilde{C}^k(\lambda,\mu)[\Delta].$
Observe that the differential $d$ preserves the conformal grading and
\[\widetilde{C}^k(\lambda,\mu)[\Delta_1]\, {}_{(n)}\,
\widetilde{C}^k(\lambda,\mu)[\Delta_2] \ \subset\ \widetilde{C}^k(\lambda,\mu)[\Delta_1+\Delta_2-n-1].
\]
Note that an endomorphism $H$ of $\widetilde{C}^k(\lambda,\mu)$ such that $H(A)=\Delta(A)A$
for any homogeneous element $A\in \widetilde{C}^k(\lambda,\mu)$ is a Hamiltonian operator.
This observation will be used in Section \ref{subsec:Zhu algebra}.

\begin{rem} \label{rem:conf}
    If a Lie algebra admits a nondegenerate bilinear form, the corresponding
    affine vertex algebra at a non-critical level contains a conformal vector.
    This is known as the {\em Sugawara construction},
    and it induces conformal vectors of the $W$-algebras (off the critical level)
     which are quantum Hamiltonian
    reductions of the affine vertex algebras.
    However, since our bilinear form \eqref{eq:bilinearform} is degenerate,
    the affine vertex algebra $V^k(\mathfrak{a})$ need not have the Sugawara
    operator, and the conformal weight \eqref{eq:Delta in W} does not come
    from a conformal vector of $W^k(\lambda, \mu)$; see also Example \ref{rem:no conformal vector}.
    \qed
\end{rem}

We will need the $\mathbb{Z}/2$-bigrading on $\widetilde{C}^k(\lambda, \mu)$ defined as follows:
\beql{eq:bigrading}
\begin{aligned}
    \text{gr}(J_a)&= \left(j_a-\frac{1}{2}, -j_a+\frac{1}{2}\right),
    \qquad \text{gr}(\partial)=(0,0),\\
    \text{gr}(\phi^{I})&= \left(-\text{deg}(E_I)+\frac{1}{2},
    \text{deg}(E_I)+\frac{1}{2} \right).
\end{aligned}
\eeq
The bigrading induces the $\mathbb{Z}_{+}$-grading on $\widetilde{C}^k(\lambda,\mu)$ given by
\ben
\widetilde{C}^k(\lambda,\mu):
= \bigoplus_{n\in \mathbb{Z}_{+}}\widetilde{C}^k(\lambda,\mu)^n,
\een
 where
\ben
    \widetilde{C}^k(\lambda,\mu)^n= \text{Span}_{\mathbb{C}} \big\{\, :a_1a_2
\dots a_s: \, |\, a_k\in R_-, \  \text{gr}(a_k)=(p_k, q_k), \ \sum_{k=1}^{s}(p_k+q_k)=n \, \big\}.
\een
It also induces the decreasing filtration
\begin{equation} \label{eq:filtration} F^p(\widetilde{C}^k(\lambda,\mu))
=\text{Span}_{\mathbb{C}} \big\{\, :a_{1}a_{2}
\dots a_{s}: \, |\, \, a_{k}\in R_-, \  \text{gr}(a_{k})=(p_k, q_k),
\ \sum_{k=1}^{s}p_k \geqslant  p \, \big\}.
\end{equation}
One easily checks that
\ben
    \widetilde{Q}(F^p \widetilde{C}^k(\lambda,\mu)^n[\Delta])
    \subset F^p \widetilde{C}^k(\lambda,\mu)^{n+1}[\Delta].
\een
The associated graded algebra is then defined by
\begin{equation} \label{eq:graded complex}
\text{gr}(\widetilde{C}^k(\lambda,\mu)):=\bigoplus_{p,q\in \mathbb{Z}/{2}}
\text{gr}^{pq} \widetilde{C}^k(\lambda,\mu),
\end{equation}
where
\ben
\text{gr}^{pq} \widetilde{C}^k(\lambda,\mu)
= F^p(\widetilde{C}^k(\lambda,\mu)^{p+q})/F^{p+\frac{1}{2}}(\widetilde{C}^k(\lambda,\mu)^{p+q})
\een
and
\ben
F^p(\widetilde{C}^k(\lambda,\mu)^{p+q})
= F^p(\widetilde{C}^k(\lambda,\mu))\cap \widetilde{C}^k(\lambda,\mu)^{p+q}.
\een
The corresponding graded cohomology is defined by
\begin{equation} \label{eq:graded cohomology}
    H^{pq}( \text{gr} (\widetilde{C}^k(\lambda,\mu)), \widetilde{Q}^{\text{gr}})
    = \frac{\text{ker}(\widetilde{Q}^{\text{gr}}:\text{gr}^{pq}\widetilde{C}^{k}
    (\lambda,\mu)\to\text{gr}^{p\, q+1}\widetilde{C}^{k}(\lambda,\mu)) }{\text{im}
    (\widetilde{Q}^{\text{gr}}:\text{gr}^{p\, q-1}\widetilde{C}^{k}(\lambda,\mu)
    \to\text{gr}^{p q}\widetilde{C}^{k}(\lambda,\mu)) },
\end{equation}
where $\widetilde{Q}^{\text{gr}}$ is the differential on
$ \text{gr}(\widetilde{C}^{k}(\lambda,\mu))$  induced from $\widetilde{Q}.$

Some structural properties of the generalized affine $W$-algebra $W^k(\lambda,\mu)$
will be derived by
using the cohomology of the graded algebra. The following
two lemmas will be essential to derive the main result of this section.
Recall the character $\chi$ as defined in \eqref{defchi} and the subspace
$\mathfrak{p}$ defined in \eqref{subp}.

\begin{lem}\label{lem:surjectivity}
    Extend the domain of the character $\chi$ to $\mathfrak{a}$
    by letting $\chi(p)=0$ for any $p\in \mathfrak{p}.$
    Then the map
    \begin{equation} \label{varphi}
    \varphi : \mathfrak{p} \to \mathfrak{a}^\ast
    \to \mathfrak{n}_{\lambda,\mu}^\ast, \qquad a
    \mapsto a \cdot \chi \mapsto (a \cdot \chi)|_{\mathfrak{n}_{\lambda, \mu}}
    \end{equation}
    is surjective.
    \begin{proof}
        If $E_{ij}^{(r)}\in \mathfrak{n}_{\lambda,\mu}$,
        then $\col_\mu(i) < \col_\mu(j)$ and
        $\lambda_j - \text{min}\{\lambda_i, \lambda_j\} \leqslant r < \lambda_j$ and
        hence $E_{j-1,i}^{(\lambda_j-r-1)}$ is an element in $\mathfrak{p}.$
        Observe that
    \begin{equation}
    \non
    \begin{aligned}
        &\varphi(E_{j-1, i}^{(\lambda_j - r - 1)})(E_{i', j'}^{(r')})
        = (E_{j-1, i}^{(\lambda_j - r - 1)} \cdot \chi) (E_{i', j'}^{(r')})
        = \chi([E_{i'j'}^{(r')}, E_{j-1, i}^{(\lambda_j-r-1)}]) \\
        &= \delta_{j', j-1} \delta_{i', i-1}
        \delta_{r', r+\lambda_i - \lambda_j}
        \delta_{\col_\mu (i-1)+1, \col_\mu (i)}
        - \delta_{i, i'} \delta_{j, j'} \delta_{r, r'} \delta_{\col_\mu (j-1)+1, \col_\mu (j)}
    \end{aligned}
    \end{equation}
    and so
    \begin{equation} \label{eq:lem_sur_pf}
    \varphi(E_{j-1, i}^{(\lambda_j - r - 1)})
    =E_{j-1, i}^{(\lambda_j-r-1)} \cdot \chi
    = (E_{i-1, j-1}^{(\lambda_i-\lambda_j+r)})^\ast
    \delta_{\col_\mu(i-1)+1, \col_\mu(i)} - (E_{ij}^{(r)})^\ast.
    \end{equation}
    Here we noted that $\col_\mu(j) \ne 1$ which
    implies $\col_\mu(j-1)+1 = \col_\mu(j)$.

    Now we show that any element $(E_{ij}^{(r)})^*$ of $\mathfrak{n}_{\lambda,\mu}^*$ is
    contained in
    $\text{Im}\, \varphi$ by the induction on $\col_\mu(i)$. If $\col_{\mu}(i) = 1$ then
    \ben
    \varphi(E_{j-1, i}^{(\lambda_j - r-1)}) = -(E_{ij}^{(r)})^\ast
    \een
    by \eqref{eq:lem_sur_pf}, i.e.  $(E_{ij}^{(r)})^\ast \in \text{Im}(\varphi)$.
    Suppose $\col_\mu(i) = c>1$. By the induction hypothesis, we have
    $(E_{i-1, j-1}^{(\lambda_i-\lambda_j+r)})^\ast\in \text{Im}\, \varphi$ and
    \ben
        \begin{aligned}
            \varphi(E_{j-1, i}^{(\lambda_j-r-1)})
            = (E_{i-1, j-1}^{(\lambda_i-\lambda_j+r)})^\ast - (E_{ij}^{(r)})^\ast\in \text{Im}\, \varphi.
        \end{aligned}
    \een
    Hence $(E_{ij}^{(r)})^\ast\in \text{Im}\, \varphi.$
    Therefore $\varphi$ is surjective.
    \end{proof}
\end{lem}

The next lemma concerns the properties of the complex $(\widetilde{C}^k(\lambda,\mu),\widetilde{Q})$.
To state it, recall that $r_-=\phi^{\mathfrak{n}_{\lambda,\mu}^*}\oplus J_{\mathfrak{p}}$
and set
\ben
r_-^{pq}[\Delta]= \text{gr}^{pq}(\widetilde{C}^k(\lambda,\mu)[\Delta]\cap r_-).
\een

\begin{lem} \label{lem:good}
In the complex $(\widetilde{C}^k(\lambda,\mu),\widetilde{Q})$ we have the following.
\begin{enumerate}[(1)]
    \item For any $\Delta\in \mathbb{Z},$ the space
    $\widetilde{C}^k(\lambda,\mu)[\Delta]$ is finite dimensional.
    \item
    We have the inclusion
    \ben
    \widetilde{Q}^{\text{gr}}(r_-^{pq}[\Delta]) \subset r_-^{p\, q+1}[\Delta],
    \een
    so that $\widetilde{Q}$ is almost a linear differential.
    \item The differential $\widetilde{Q}$ on $\widetilde{C}^k(\lambda,\mu)$
    is a good differential; that is,
    \ben
    H^{pq}(\text{gr}(\widetilde{C}^{k}(\lambda, \mu))[\Delta], \widetilde{Q}^{\text{gr}})=0
    \een
    unless $p+q=0.$
\end{enumerate}
\end{lem}

\begin{proof}
    Part (1) is clear, while parts (2) and (3) follow directly from the fact that
    \ben
        \widetilde{Q}^{\text{gr}}(J_a)
        =-\phi^{a \cdot \chi}\Fand \widetilde{Q}^{\text{gr}}(\phi^{\mathsf{m}})=0
    \een
    for $a\in \mathfrak{p}$ and
    $\mathsf{m}\in \mathfrak{n}^\ast_{\lambda, \mu}$, along with Lemma \ref{lem:surjectivity}.
\end{proof}

\begin{rem}\label{rem:chi}
    The choice of \eqref{defchi} was made so as to ensure Lemma~\ref{lem:surjectivity} and Lemma~\ref{lem:good} (3). Let $\chi : \mathfrak{n}_{\lambda, \mu} \to \mathbb{C}$ be an arbitrary character. The condition $\chi(E_{i, i+1}^{(\lambda_{i+1}-1)}) \ne 0$ is necessary for Lemma~\ref{lem:surjectivity}. If one chooses $\chi$ different from \eqref{defchi} but still satisfying this condition, then Lemma~\ref{lem:good} (3) and Theorem~\ref{thm:degree zero} remain valid.
\end{rem}

Invoking
\cite[Lemma 4.18 and Theorem 4.19]{dk:fa},
we come to the following theorem.

\begin{thm} \label{thm:degree zero}
Let  $\{e_i \ |\ i \in S_{\chi}\}$ be a basis of
$r_-^{p,-p}[\Delta] \cap \text{ker\,} \widetilde{Q}^{\text{gr}}$ and
let
\ben
E_i= e_i+ \epsilon _i\in F^p(\widetilde{C}^k(\lambda,\mu)^0)[\Delta] \cap \text{ker\,} \widetilde{Q}
\een
for some $\epsilon_i\in F^{p+\frac{1}{2}} \widetilde{C}^{k}(\lambda, \mu)^0[\Delta].$ Then
\[H(R_-,\widetilde{Q}):= \mathbb{C}[\partial]\otimes
\text{Span}_{\mathbb{C}}\{E_i|i\in S_\chi\}\]
is a nonlinear Lie conformal algebra. Moreover,
\[H(\widetilde{C}^k(\lambda,\mu),\widetilde{Q})
=H^0(\widetilde{C}^k(\lambda,\mu),\widetilde{Q}) \simeq V(H(R_-,\widetilde{Q}))\]
where $V(H(R_-,\widetilde{Q}))$ is the
universal enveloping vertex algebra of $H(R_-,\widetilde{Q}).$
\qed
\end{thm}

As a consequence of
Theorem \ref{thm:degree zero}, we get the following property of the
generalized affine $W$-algebra:
\begin{equation} \label{eq:affine and W}
    W^k(\lambda,\mu)\ \  {\text{\em is a vertex subalgebra of }}\ \  V(J_{\mathfrak{p}}),
\end{equation}
where $V(J_{\mathfrak{p}})$ is the universal enveloping vertex algebra of
$\mathbb{C}[\partial]\otimes J_{\mathfrak{p}}$ endowed with the
$\lambda$-bracket introduced in \eqref{eq: two building block}.
Take an ordered basis $\mathcal{B}_{\mathfrak{p}}$ of $J_{\mathfrak{p}}$.
Then $V(J_{\mathfrak{p}})$ has the PBW basis induced from the ordered
set $\wt{\mathcal{B}}= \{\, \partial^t \mathcal{B}_{\mathfrak{p}}\, |\, t\in \mathbb{Z}_+\, \}$
and we can define degrees of elements of $V(J_{\mathfrak{p}})$ by
degrees with respect to the basis elements in $\wt{\mathcal{B}}$.
In particular, the {\em linear term} of any element will mean the degree $1$ part in
        $V(J_{\mathfrak{p}})$ with respect to the given basis of $J_{\mathfrak{p}}$.

\begin{cor} \label{cor:affine structure}
    Let $a_1, \dots, a_r$ be a basis of $\text{ker}\, \varphi$ for the map $\varphi$ in \eqref{varphi}.
\begin{enumerate}[(1)]
    \item  As a differential algebra, $W^k(\lambda,\mu)$ has
    a generating set consisting of $r=\text{dim}\, \mathfrak{a}(0)$ elements.
    Moreover, $r$ is the minimal number of elements in a generating set.
    \item There exists a differential algebraically independent generating
    set $\{w_i\ |\ i=1,\dots, r\}$ of $W^k(\lambda,\mu)$ satisfying the following properties:
    \begin{enumerate}[(i)]
        \item The element $w_i$ is homogeneous with respect to the conformal weight.
        \item The linear term of $w_i$ without total derivative is $J_{a_i}$,
        with respect to any basis of $J_{\mathfrak{p}}$.

    \end{enumerate}
    Moreover, any subset  $\{w_i\ |\ i=1,\dots, r\}\subset W^k(\lambda,\mu)$
    with properties (i) and (ii), is a generating set of $W^k(\lambda,\mu)$.
\end{enumerate}
\end{cor}

\begin{proof}
    By Theorem \ref{thm:degree zero}, the number of free generators
    of $W^k(\lambda,\mu)$ is $\text{dim}(\text{ker }\varphi)$. Lemma~\ref{lem:surjectivity}
    implies that
    \ben
    \text{dim}(\text{ker }\varphi) = \text{dim }\mathfrak{p} - \text{dim }
    \mathfrak{n}_{\lambda, \mu}^\ast = \text{dim } \mathfrak{p}
    - \text{dim } \mathfrak{n}_{\lambda, \mu} = \text{dim } \mathfrak{a}(0).
    \een
    Hence part (1) holds. Part (2) is a direct
    consequence of Theorem \ref{thm:degree zero}.
\end{proof}

\begin{rem}
\label{rem:krw}
\begin{enumerate}[(a)]
    \item In the particular case $\lambda=(1^N)$,
our BRST complex defining the $W$-algebra $W^k(\lambda,\mu)$
coincides with the one used in \cite[Sec.~5]{dk:fa} to describe the $W$-algebra $W^k(\gl_N,f)$
associated with a nilpotent element $f\in\gl_N$ of type $\mu$. Therefore,
these two affine $W$-algebras are the same.
\item For other types of Lie algebras, similar construction should work. The only nontrivial part is finding various gradations on $\mathfrak{a}=\g^e,$ which substitute the partition $\mu$ in $\g=\mathfrak{gl}_N$ case.
\end{enumerate}

\end{rem}

\section{Application of the Zhu functor}

\subsection{Zhu algebra of $W^k(\lambda,\mu)$} \label{subsec:Zhu algebra}
In this section, we obtain the Zhu algebra of generalized affine $W$-algebra
which will be proved to be isomorphic to $U(\lambda,\mu).$
We start by recalling basic definitions related to the Zhu functor; see e.g. \cite{dk:fa}
for a more detailed discussion.

An operator $H$
on a vertex algebra $V$
is called a {\em Hamiltonian} if it is a diagonalizable operator satisfying the property:
\begin{equation} \label{eq:hamiltonian}
    \Delta(a_{(n)}b)= \Delta(a)+\Delta(b)-n-1,
\end{equation}
where $a$ and $b$ are eigenvectors for $H$ and $H(c)=:\Delta(c)c$
for any eigenvector $c$. The $H$-{\em twisted Zhu algebra} $Zhu_H(V)$ of $V$
is the associative algebra given by
\[ Zhu_H(V)= V/(\partial+H)V.\]
The associative algebra product and the commutator relation on $Zhu_H(V)$ are given by
\begin{equation} \label{eq:Zhu_product}
    \begin{aligned}
        & Zhu_H(a)\cdot Zhu_H(b)= Zhu_H(:ab:) + \int^{1}_0  [Zhu_H(H(a)),Zhu_H(b)]_x dx,\\
        & [Zhu_H(a),Zhu_H(b)]:= [Zhu_H(a),Zhu_H(b)]_{\hbar=1}
    \end{aligned}
\end{equation}
where
\ben
[Zhu_H(a),Zhu_H(b)]_\hbar
=\sum_{j\in \mathbb{Z}_+} { \Delta_a - 1 \choose j} \,
\hbar^j\, Zhu_H(a_{(j)} b).
\een
Note that the conformal weight  \eqref{eq:Delta in W}
on the complex $\widetilde{C}^k(\lambda,\mu)$ has the property \eqref{eq:hamiltonian}.
The conformal weight \eqref{eq:Delta in W} can also be extended to $C^k(\lambda,\mu)$ by letting
\begin{equation} \label{eq:conformal grading on C}
    \Delta(a)=1-j_a, \quad  \Delta(\phi_{\mathsf{n}})=1-j_{\mathsf{n}}
\end{equation}
for $a\in \mathfrak{a}(j_a),$ and $\mathsf{n}\in \mathfrak{a}(j_\mathsf{n})$.
Consider the Hamiltonian operator $H$ on $C^k(\lambda,\mu)$ defined
by \eqref{eq:conformal grading on C} and let \[C^{\text{fin}}(\lambda,\mu):=Zhu_H(C^k(\lambda,\mu)).\]
We set
\begin{equation}
\non
    \bar{a}:=Zhu_H(a)-k(h|a), \quad \bar{\phi}^{\mathsf{m}}:=Zhu_H(\phi^{\mathsf{m}}),
    \quad \bar{\phi}_{\mathsf{n}}:=Zhu_H(\phi_{\mathsf{n}}), \quad \bar{J_a} := Zhu_H(J_a),
\end{equation}
where \[h=\sum_{\substack{i=1,\dots, n \\ \col_\mu(i)=r}}\ (n-r) \cdot E_{ii}^{(0)}.\]
Note that $[h, a] = j_a a$ for $a\in \mathfrak{a}(j_a)$. In other words,
the grading $\bigoplus_{m\in \mathbb{Z}}\mathfrak{a}(m)$ can be induced from
the Dynkin grading by $h$ on $\mathfrak{gl}_N$.
Using \eqref{eq:Zhu_product}, we obtain
\ben
[\bar{a}, \bar{b}] = [Zhu_H(a), Zhu_H(b)] = Zhu_H([a, b])
+ (\Delta_a -1)(a|b)k = Zhu_H([a, b]) - j_a(a|b)k
\een
and hence
\ben
[\bar{a},\bar{b}]=\overline{[a,b]}, \quad [\bar{\phi}^{\mathsf{m}}, \bar{\phi}_{\mathsf{n}}]
=\mathsf{m}(\mathsf{n}).
\een
Let us consider the differential $d^{\text{fin}} := Zhu_H(d)
\in C^{\text{fin}}(\lambda,\mu)$ and $Q^{\text{fin}} = \text{ad } d^{\text{fin}}.$ Then we have
    \begin{align}
     Q^{\text{fin}}(\bar{a}) &= Zhu_H(d_{(0)} a) \label{degone}\\[0.4em]
     {} &= Zhu_H(\phi^{I}) Zhu_H([E_I, a])
     - \sum_{I \in S_{\lambda, \mu}} k([h, E_I] | a) Zhu_H \phi^I
     = \sum_{I\in S_{\lambda, \mu}} \bar{\phi}^{I}  \overline{[E_{I}, a]}, \non\\
     Q^{\text{fin}} (\bar{\phi}_{\mathsf{n}}) &= \bar{\mathsf{n}}
     + \chi(\mathsf{n}) + \sum_{I\in S_{\lambda, \mu}} \bar{\phi}^{I}
     \bar{\phi}_{[E_{I}, \mathsf{n}]}, \ \  Q^{\text{fin}}(\bar{\phi}^{\mathsf{m}})
     = \frac{1}{2} \sum_{I \in S_{\lambda,\mu}}
     \bar{\phi}^{I} \bar{\phi}^{E_{I} \cdot \mathsf{m}},\non\\
     Q^{\text{fin}}(\bar{J}_a) &= \bar{\phi}^I \bar{J}_{\pi_{\le} [E_I, a]}
     - \bar{\phi}^{a \cdot \chi} + \bar{\phi}^{\pi_{\le} [a, [h, E_I]] \cdot E_I^\ast} \non\\
     & - \sum_{I \in S_{\lambda, \mu}} \Big(k([h, E_I] | a)
     + \text{tr}_{\mathfrak{n}_{\lambda, \mu}} ((\pi_+ \text{ad } [h, E_I])
     \circ (\pi_+ \text{ad } a))\Big) \bar{\phi}^I\, .
     \non
    \end{align}
Note that equality \eqref{degone}
holds because the conformal weight
of the element $d$ is equal to $1$.
Consider the decomposition
\[C^{\text{fin}}(\lambda,\mu)
=C^{\text{fin}}_+(\lambda, \mu)\otimes \widetilde{C}^{\text{fin}}(\lambda, \mu),\]
where $C^{\text{fin}}_+(\lambda, \mu)$ and $\widetilde{C}^{\text{fin}}(\lambda, \mu)$ are
subalgebras of $C^{\text{fin}}(\lambda, \mu)$ respectively generated by
the subspaces
$\bar{\phi}_{\mathfrak{n}_{\lambda, \mu}} \oplus \bar{J}_{\mathfrak{n}_{\lambda, \mu}}$
and $\bar{\phi}^{\mathfrak{n}_{\lambda, \mu}^\ast} \oplus \bar{J}_{\mathfrak{p}}$.
Then the operator $H$ can be regarded as a Hamiltonian operator on both
$C^k_+(\lambda, \mu)$ and $\widetilde{C}^k(\lambda, \mu)$ and
\[Zhu_H(C^k_+(\lambda, \mu)) = C^{\text{fin}}_+(\lambda, \mu),
\quad Zhu_H(\widetilde{C}^k(\lambda, \mu)) = \widetilde{C}^{\text{fin}}(\lambda, \mu).\]
Similarly to the affine case, the differential $Q^{\text{fin}}$
satisfies
\ben
Q^{\text{fin}} (C^{\text{fin}}_+ (\lambda, \mu)) \subset C^{\text{fin}}_+ (\lambda, \mu),\qquad
\bar{Q}(\widetilde{C}^{\text{fin}}(\lambda, \mu)) \subset \widetilde{C}^{\text{fin}}(\lambda, \mu)
\een
and
\ben
H(C^{\text{fin}}_+(\lambda, \mu), Q^{\text{fin}}|_{C^{\text{fin}}_+(\lambda, \mu)}) = \mathbb{C}.
\een
Therefore,
\begin{equation}
\non
    H(C^{\text{fin}}(\lambda, \mu), \bar{Q})
    \cong H(C^{\text{fin}}_+(\lambda, \mu), \bar{Q}|_{C^{\text{fin}}_+(\lambda, \mu)})
    \otimes H(\widetilde{C}^{\text{fin}}(\lambda, \mu), \widetilde{Q}^{\text{fin}})
    \cong H(\widetilde{C}^{\text{fin}}(\lambda, \mu), \widetilde{Q}^{\text{fin}}),
\end{equation}
where $\widetilde{Q}^{\text{fin}} := Q^{\text{fin}}|_{\widetilde{C}^{\text{fin}}(\lambda, \mu)}$.

On the other hand, since the Hamiltonian operator $H$ on $C^k(\lambda,\mu)$
induces the Hamiltonian operator on the affine $W$-algebra $W^k(\lambda,\mu)$,
we get the associative algebra $Zhu_H(W^k(\lambda,\mu)).$ By
applying \cite[Theorem 4.20]{dk:fa}, we come to the following proposition.

\begin{prop} \label{prop:Zhu and coho commute}
    The following two associative algebras are isomorphic:
    \[ Zhu_H(W^k(\lambda,\mu))\simeq H(C^{\text{\rm fin}}(\lambda,\mu),Q^{\text{\rm fin}}).\]
\end{prop}
\begin{proof}
    We have the grading \eqref{eq:bigrading} and filtration \eqref{eq:filtration} on
    $\widetilde{C}^k(\lambda, \mu)$ satisfying
    \ben
    F^{p_1}\widetilde{C}^k(\lambda, \mu)^{n_1}[\Delta_1]_{(n)}
    F^{p_2}\widetilde{C}^k(\lambda, \mu)^{n_2}[\Delta_2] \subset F^{p_1+p_2}
    \widetilde{C}^k(\lambda, \mu)^{n_1+n_2}[\Delta_1+\Delta_2-n-1]
    \een
    and
    \ben
    \widetilde{Q}(F^{p}\widetilde{C}^k(\lambda, \mu)^n[\Delta])
    \subset F^p\widetilde{C}^k(\lambda, \mu)^{n+1}[\Delta].
    \een
    Furthermore, due to Lemma \ref{lem:good}, the assumptions
    of \cite[Theorem~4.20]{dk:fa} are satisfied. Thus, there is a
    canonical associative algebra isomorphism
    \ben
    Zhu_H (W^k(\lambda, \mu))
    = Zhu_H H(\widetilde{C}^k(\lambda, \mu), \widetilde{Q})
    \cong H(Zhu_H(\widetilde{C}^k(\lambda, \mu)), \widetilde{Q})
    = H(\widetilde{C}^{\text{fin}}(\lambda, \mu), \widetilde{Q}^{\text{fin}}),
    \een
    completing the proof.
\end{proof}

\subsection{The algebras $U(\lambda,\mu)$ and $Zhu_H W^k(\lambda,\mu)$ } \label{eq:equiv_finite W}

We will show that the Zhu algebra of $W^k(\lambda,\mu)$ is
isomorphic to the generalized finite $W$-algebra $U(\lambda,\mu)$ by
analyzing the graded algebra of $U(\lambda,\mu)$.
To this end, we need to understand $U(\lambda,\mu)$ via
Lie algebra cohomology. As the first step,
construct the classical finite version of BRST complex.
Denote $S(V)$ the supersymmetric algebra of a vector superspace $V$ and set
\begin{equation}
\non
    \mathcal{C}(\lambda, \mu) := S(\mathfrak{a})
    \otimes S(\phi_{\mathfrak{n}_{\lambda, \mu}} \oplus \phi^{\mathfrak{n}_{\lambda, \mu}^\ast}),
\end{equation}
where $\phi_{\mathfrak{n}_{\lambda, \mu}}$, $\phi^{\mathfrak{n}_{\lambda, \mu}^\ast}$
are the odd vector superspaces defined in \eqref{eq:phi}.
The Poisson bracket on $\mathcal{C}(\lambda,\mu)$ is induced
from the Lie bracket on $\mathfrak{a}$ and
the super-commutator $[\phi^{\mathsf{m}}, \phi_{\mathsf{n}}]=\mathsf{m}(\mathsf{n}).$
Consider the element
\ben
    d^{cl} := \sum_{I \in S_{\lambda, \mu}} \phi^{I}E_{I}
    + \phi^{\chi} + \frac{1}{2} \sum_{I,I'\in S_{\lambda, \mu}}
    \phi^{I} \phi^{I'} \phi_{[E_{I'}, E_{I}]} \in \mathcal{C}(\lambda, \mu)
\een
and set
\ben
    \mathcal{Q} := \text{ad } d^{cl} = \{d^{cl}, \cdot \} \in \text{End}\ts\mathcal{C}(\lambda, \mu).
\een
By direct computation, one can show  $\mathcal{Q}^2 =0$ so that $\mathcal{Q}$
is a differential of the complex $\mathcal{C}(\lambda,\mu).$
For $a\in \mathfrak{a}$ set
\begin{equation}
\non
    J^{cl}_a = a + \sum_{I \in S_{\lambda, \mu}} \phi^{I}
    \phi_{[E_{I}, a]}\in \mathcal{C}(\lambda,\mu)
\end{equation}
and introduce the spaces
\ben
J^{cl}_{\mathfrak{n}_{\lambda,\mu}}:=
\{ J^{cl}_{\mathsf{n}}|\mathsf{n}\in \mathfrak{n}_{\lambda,\mu}\},\qquad
J^{cl}_{\mathfrak{p}}:= \{ J^{cl}_{p}|p\in \mathfrak{p}\}.
\een
Then
$\mathcal{C}_+(\lambda, \mu) := S(\phi_{\mathfrak{n}_{\lambda, \mu}}
\oplus  J^{cl}_{\mathsf{n}_{\lambda, \mu}})$ and
$\widetilde{\mathcal{C}}(\lambda, \mu)  := S( J^{cl}_{\mathfrak{p}}
\oplus \phi^{\mathfrak{n}_{\lambda, \mu}^\ast})
$
are Poisson subalgebras of $\mathcal{C}(\lambda, \mu)$.
As in the affine case, we can show that
\[ \mathcal{Q}(\mathcal{C}_+(\lambda, \mu)) \subset \mathcal{C}_+(\lambda, \mu),
\quad \mathcal{Q}(\widetilde{\mathcal{C}}(\lambda, \mu)) \subset \widetilde{\mathcal{C}}(\lambda, \mu).\]
Moreover, since $\mathcal{C}(\lambda,\mu)\cong
\mathcal{C}_+(\lambda, \mu)\otimes \widetilde{\mathcal{C}}(\lambda, \mu)$
and $H(\mathcal{C}_+(\lambda, \mu), \mathcal{Q}|_{\mathcal{C}_+(\lambda,\mu}))=\mathbb{C}$, we have
\begin{equation} \label{eq:finite cl cohomology}
    H(\mathcal{C}(\lambda, \mu), \mathcal{Q})
    \cong H(\widetilde{\mathcal{C}}(\lambda, \mu),  \widetilde{\mathcal{Q}}),
\end{equation}
where $\widetilde{\mathcal{Q}}:=\mathcal{Q}|_{\widetilde{\mathcal{C}}(\lambda,\mu)}.$
The differential $\widetilde{\mathcal{Q}}$ acts on the elements in
$\widetilde{\mathcal{C}}(\lambda,\mu)$ as follows:
\begin{equation} \label{eq:tilde of Q}
 \widetilde{\mathcal{Q}} (J^{cl}_a)=
 \sum_{I \in S_{\lambda, \mu}} \phi^{I} J^{cl}_{\pi_\le[E_{I}, a]} - \phi^{a \cdot \chi},
\qquad
\widetilde{\mathcal{Q}} (\phi^{\mathsf{m}})
= \frac{1}{2} \sum_{I \in S_{\lambda,\mu}}  \phi^{I} \phi^{E_{I} \cdot \mathsf{m}}.
\end{equation}

Similar to the
conformal weight on the complex $\widetilde{C}^k(\lambda,\mu)$,
we define the $\Delta$-grading on $\widetilde{\mathcal{C}}(\lambda,\mu)$
and call it {\em conformal weight},
by setting
\[ \Delta(J^{cl}_a) = 1-j_a, \qquad \Delta(\phi^{I}) = \text{deg}(E_{I})\]
for $a\in \mathfrak{a}(j_a)$
and introduce the space
\begin{equation} \label{eq:Delta-filtration}
    \widetilde{\mathcal{C}}(\lambda,\mu)_{\Delta}
    = \text{Span}_{\mathbb{C}}\{A\in \widetilde{\mathcal{C}}(\lambda,\mu)|\Delta(A)\leqslant \Delta \}.
\end{equation}
We also consider the $\mathbb{Z}/2$-bigrading
\ben
\text{gr}(J^{cl}_a) = \left( j_a-\frac{1}{2}, -j_a+\frac{1}{2} \right),
\quad \text{gr}(\phi^{I}) = \left( -\text{deg}(E_{I}) + \frac{1}{2},
\text{deg}(E_{I}) + \frac{1}{2} \right)
\een
and the induced $\mathbb{Z}_+$-grading given by
\ben
\widetilde{\mathcal{C}}(\lambda, \mu) = \bigoplus_{n \in \mathbb{Z}_+}
\widetilde{\mathcal{C}}(\lambda, \mu)^n,
\een
where
\begin{equation}
\non
    \widetilde{\mathcal{C}}(\lambda, \mu)^n
    = \text{Span}_\mathbb{C} \big\{ a_1a_2\dots a_s \ | \ a_k \in r_-,
    \text{gr}(a_k) = (p_k, q_k), \sum_{k=1}^s (p_k + q_k) = n \big\}
\end{equation}
for
\beql{rminus}
r_-= J^{cl}_{\mathfrak{p}}\oplus \phi^{\mathfrak{n}^\ast_{\lambda,\mu}}.
\eeq
Using the decreasing filtration
\begin{equation}
\non
    F^p\widetilde{\mathcal{C}}(\lambda, \mu)
    = \text{Span}_{\mathbb{C}} \big\{a_1a_2 \dots a_s \ |
    \ a_k \in r_-, \text{gr}(a_k) = (p_k, q_k) , \sum_{k=1}^s p_k \geqslant p\big\},
\end{equation}
define the associated graded algebra by
\begin{equation}
\non
\text{gr}(\widetilde{\mathcal{C}}(\lambda,\mu)):=\bigoplus_{p,q\in
\mathbb{Z}/{2}} \text{gr}^{pq} \widetilde{\mathcal{C}}(\lambda,\mu)
\end{equation}
where
\ben
\text{gr}^{pq} \widetilde{\mathcal{C}}(\lambda,\mu)
= F^p\widetilde{\mathcal{C}}(\lambda,\mu)^{p+q})/F^{p+\frac{1}{2}}
\widetilde{\mathcal{C}}(\lambda,\mu)^{p+q}
\een
and
\ben
F^p\widetilde{\mathcal{C}}(\lambda,\mu)^{p+q}= F^p\widetilde{\mathcal{C}}
(\lambda,\mu)\cap \widetilde{\mathcal{C}}(\lambda,\mu)^{p+q}.
\een
The corresponding graded cohomology is defined by
\begin{equation}
\non
    H^{pq}( \text{gr} (\widetilde{\mathcal{C}}(\lambda,\mu)), \widetilde{\mathcal{Q}}^{\text{gr}})
    = \frac{\text{ker}(\widetilde{\mathcal{Q}}^{\text{gr}}:\text{gr}^{pq}
    \widetilde{\mathcal{C}}(\lambda,\mu)\to\text{gr}^{p\, q+1}
    \widetilde{\mathcal{C}}(\lambda,\mu)) }{\text{im}
    (\widetilde{\mathcal{Q}}^{\text{gr}}:\text{gr}^{p\, q-1}\widetilde{\mathcal{C}}(\lambda,\mu)
    \to\text{gr}^{p q}\widetilde{\mathcal{C}}(\lambda,\mu)) },
\end{equation}
where $\widetilde{\mathcal{Q}}^{\text{gr}}$ is the differential
on $ \text{gr}(\widetilde{\mathcal{C}}(\lambda,\mu))$  induced from $\widetilde{\mathcal{Q}}.$

\if
On the other hand, the cohomology \eqref{eq:finite cl cohomology}
can be obtained as a graded cohomology of
 $\widetilde{C}^{\text{fin}}(\lambda, \mu)$ in Section \ref{subsec:Zhu algebra}.
 More precisely, we define the $\Delta$-grading
\ben
    \Delta(J_a) = 1-j_a, \qquad \Delta(\phi^{(i, j, r)}) = \text{deg}(E_{ij}^{(r)}),
\een
and the $\mathbb{Z}/2$-bigrading
\ben
    \text{gr}(J_a) = \left( j_a-\frac{1}{2}, -j_a+\frac{1}{2} \right),
    \quad \text{gr}(\phi^{(i, j, r)}) = \left( \text{deg}(E_{ij}^{(r)})
    + \frac{1}{2}, -\text{deg}(E_{ij}^{(r)}) + \frac{1}{2} \right).
\een
Then the bigrading induces the $\mathbb{Z}_+$ grading
$\widetilde{C}^{\text{fin}}(\lambda, \mu) = \oplus_{n \in \mathbb{Z}_+}
\widetilde{C}^{\text{fin}}(\lambda, \mu)^n$ where
\ben
    \widetilde{C}^{\text{fin}}(\lambda, \mu)^n
    = \text{Span}_\mathbb{C} \big\{ a_1a_2\dots a_s \ | \ a_k \in r_-, \text{gr}(a_k)
    = (p_k, q_k), \sum_{k=1}^s (p_k + q_k) = n \big\}
\een
and the decreasing filtration
\ben
    F^p(\widetilde{C}^{\text{fin}}(\lambda, \mu))
    = \text{Span}_{\mathbb{C}} \big\{a_1a_2 \dots a_s \ |
    \ a_k \in r_-, \text{gr}(a_k) = (p_k, q_k) , \sum_{k=1}^s p_k \geqslant p\big\}.
\een
With respect to the filtration, we consider the graded algebra
\ben
\text{gr}(\widetilde{C}^{\text{fin}}(\lambda,\mu)):=\bigoplus_{p,q\in \mathbb{Z}/{2}}
\text{gr}^{pq} \widetilde{C}^{\text{fin}}(\lambda,\mu),
\een
where $\text{gr}^{pq} \widetilde{C}^{\text{fin}}(\lambda,\mu)
= F^p(\widetilde{C}^{\text{fin}}(\lambda,\mu)^{p+q})/
F^{p+\frac{1}{2}}(\widetilde{C}^{\text{fin}}(\lambda,\mu)^{p+q})$ and $
F^p(\widetilde{C}^{\text{fin}}(\lambda,\mu)^{p+q})
= F^p(\widetilde{C}^{\text{fin}}(\lambda,\mu))\cap \widetilde{C}^{\text{fin}}(\lambda,\mu)^{p+q}$.
The corresponding graded cohomology is defined by
\ben
    H^{pq}( \text{gr} (\widetilde{C}^{\text{fin}}(\lambda,\mu)), \bar{d}^{\text{gr}})
    = \frac{\text{ker}(\bar{d}^{\text{gr}}:\text{gr}^{pq}
    \widetilde{C}^{\text{fin}}(\lambda,\mu)\to\text{gr}^{p\, q+1}
    \widetilde{C}^{\text{fin}}(\lambda,\mu)) }{\text{im}(\bar{d}^{\text{gr}}:\text{gr}^{p\, q-1}
    \widetilde{C}^{\text{fin}}(\lambda,\mu)\to\text{gr}^{p q}\widetilde{C}^{\text{fin}}(\lambda,\mu)) },
\een
where $\bar{d}^{\text{gr}}$ is the differential on
$ \text{gr}(\widetilde{C}^{\text{fin}}(\lambda,\mu))$  induced from $\bar{d}.$
Then one can check that
\[ \text{gr} (\widetilde{C}^{\text{fin}}(\lambda, \mu)) = \widetilde{\mathcal{C}}(\lambda, \mu). \]
\fi

\begin{lem}  \label{lem:classical finite}
The following properties hold.
    \begin{enumerate}[(1)]
        \item The differential $\widetilde{\mathcal{Q}}$ on
        $\widetilde{\mathcal{C}}(\lambda,\mu)$ is good and
        $\text{\rm gr}^{pq}H(\widetilde{\mathcal{C}}(\lambda, \mu),
        \widetilde{\mathcal{Q}}) \cong H^{pq}(\text{\rm gr}(\widetilde{\mathcal{C}}(\lambda, \mu)),
        \widetilde{\mathcal{Q}}^{\text{\rm gr}})$.
        \item $H(\widetilde{\mathcal{C}}(\lambda, \mu), \widetilde{\mathcal{Q}})
        = H^{0}(\widetilde{\mathcal{C}}(\lambda, \mu), \widetilde{\mathcal{Q}})$.
    \end{enumerate}
    \end{lem}

    \begin{proof}
        \begin{enumerate}[(1)]
            \item Similarly to Lemma \ref{lem:good} (3), we can show that
            $H^{pq}(\text{gr}(\widetilde{\mathcal{C}}(\lambda, \mu)),
            \widetilde{\mathcal{Q}}^{gr}) = 0$ unless $p+q=0$,
            which means $\mathcal{Q}$ is a good differential.
            Moreover, since $\mathcal{Q}$ preserves the conformal
            weight and the subspace of $\widetilde{C}(\lambda,\mu)$
            with a given conformal weight is finite dimensional,
            we can apply \cite[Lemma 4.2]{dk:fa}.
            \item
            As in Lemma \ref{lem:good} (2), the differential
            $\widetilde{\mathcal{Q}}$ is almost linear.
            Moreover, recalling
            notation \eqref{rminus}, observe that since
            $H^{pq}(r_-, \widetilde{\mathcal{Q}}^{\text{gr}})$ is nonzero
            only when $p+q=0,$ the cohomology
            $H(\text{gr}(\widetilde{\mathcal{C}}(\lambda,\mu)), \widetilde{\mathcal{Q}}^{\text{gr}})\simeq
            S(H(r_-,\widetilde{\mathcal{Q}}^{\text{gr}}))$ is also concentrated at degree $0$ part.
            Now, by part (1) we get part (2).
            \end{enumerate}
    \end{proof}

\begin{cor}
    The cohomology $H(\widetilde{\mathcal{C}}(\lambda, \mu), \widetilde{\mathcal{Q}})$
    is a Poisson subalgebra of $S(\mathfrak{p}).$
\end{cor}
\begin{proof}
    By the same argument as in the proof of Lemma \ref{prop:BRST VA structure}, we can
    show $H(\mathcal{C}(\lambda, \mu), \mathcal{Q})$ has the Poisson algebra
    structure induced from that of $\mathcal{C}(\lambda,\mu)$.
    Since $H(\mathcal{C}(\lambda, \mu), \mathcal{Q})\simeq
    H^0(\widetilde{\mathcal{C}}(\lambda, \mu), \widetilde{\mathcal{Q}}),$
    this is a Poisson subalgebra of $S(J_{\mathfrak{p}}^{cl}).$
    Finally, since $J_{\mathfrak{p}}^{cl}\simeq \mathfrak{p}$ as Lie algebras, the proof
    is complete.
\end{proof}

Now we want to show that the Poisson algebra
$H(\widetilde{\mathcal{C}}(\lambda, \mu), \widetilde{\mathcal{Q}})$
can also be realized in two different ways:
\begin{equation} \label{eq:cl-fin-W, equiv}
    \text{(1) Lie algebra cohomology of $\mathfrak{n}_{\lambda,\mu}$;
    \qquad (2) graded algebra of $U(\lambda,\mu).$}
\end{equation}
Consider the Kazhdan filtration
$K_0(\mathfrak{a}) \subset K_1(\mathfrak{a}) \subset K_2(\mathfrak{a}) \subset \dots$ on
$U(\mathfrak{a})$, where a homogeneous element $a \in \mathfrak{a}$ is in
$ K_s(\mathfrak{a})$ if and only if $a\in \mathfrak{a}(j_a)$ for some $j_a \geqslant  1-s.$
Then for elements $a\in K_s(\mathfrak{a})$ and $b\in K_t(\mathfrak{a}),$ the commutator
$ab-ba$ belongs to $K_{s+t-1}(\mathfrak{a}).$ Then the graded algebra
$S(\mathfrak{a}):=\text{gr}^K(U(\mathfrak{a}))$ is the Poisson
algebra endowed with the bracket
\ben
\{\text{gr}^K_s(a), \text{gr}^K_t(b)\}
=\text{gr}^K_{s+t-1}(ab-ba).
\een
In addition, the Kazhdan filtration induces the
filtration on $U(\mathfrak{a})/\mathcal{I}_{\lambda,\mu}$ and
\ben
\text{gr}^K(U(\mathfrak{a})/\mathcal{I}_{\lambda,\mu})
= S(\mathfrak{a})/\mathcal{I}^{cl}_{\lambda,\mu},\qquad
\mathcal{I}^{cl}_{\lambda,\mu}:=S(\mathfrak{a})\left< n+\chi(n)|n\in \mathfrak{n}_{\lambda,\mu} \right>.
\een
Then $U(\mathfrak{a})/\mathcal{I}_{\lambda,\mu}$ and $S(\mathfrak{a})/\mathcal{I}^{cl}_{\lambda,\mu}$
are  $\mathfrak{n}_{\lambda,\mu}$-modules via adjoint actions. Here we note that the
Kazhdan filtration is the analogue of the filtration \eqref{eq:Delta-filtration}
induced from the conformal weight.
In the following proposition, we describe the two realizations
of $H(\widetilde{\mathcal{C}}(\lambda,\mu), \widetilde{Q})$ in \eqref{eq:cl-fin-W, equiv}
more precisely via the $\mathfrak{n}_{\lambda,\mu}$-module $S(\mathfrak{a})/\mathcal{I}^{cl}_{\lambda,\mu}$.

\begin{prop} \label{prop:Lie algebra coho-finite W-pre}
    The following two cohomologies are isomorphic:
        \[
        \non
        H(\widetilde{\mathcal{C}}(\lambda, \mu), \widetilde{\mathcal{Q}})
        \simeq H(\, \mathfrak{n}_{\lambda,\mu}  ,
        S(\mathfrak{a})/\mathcal{I}^{cl}_{\lambda,\mu} \, )
        = H^0(\, \mathfrak{n}_{\lambda,\mu}, S(\mathfrak{a})/\mathcal{I}^{cl}_{\lambda,\mu}\, ).\]
    Hence $H(\widetilde{\mathcal{C}}(\lambda, \mu), \widetilde{\mathcal{Q}})\simeq
    (S(\mathfrak{a})/\mathcal{I}^{cl}_{\lambda,\mu})^{\text{ad}\,  \mathfrak{n}_{\lambda,\mu}}.$
    \end{prop}
\begin{proof}
We show that the Chevalley--Eilenberg complex of the
$\mathfrak{n}_{\lambda, \mu}$-module $S(\mathfrak{a})/\mathcal{I}^{cl}_{\lambda, \mu}$
is isomorphic to the complex $(\widetilde{\mathcal{C}}(\lambda, \mu), \widetilde{\mathcal{Q}})$.
Consequently, the cohomology of $(\widetilde{\mathcal{C}}(\lambda, \mu), \widetilde{\mathcal{Q}})$
is concentrated in degree $0$, so the cohomology $H(\mathfrak{n}_{\lambda, \mu},
S(\mathfrak{a})/\mathcal{I}^{cl}_{\lambda, \mu})$ is also concentrated in degree $0$.
By \eqref{eq:tilde of Q},  the map between the complexes
    \ben
        \begin{aligned}
            \wedge^n (\mathfrak{n}_{\lambda, \mu}^\ast) \otimes
            S(\mathfrak{a})/\mathcal{I}^{cl}_{\lambda, \mu} &\to
            S(\phi^{\mathfrak{n}_{\lambda, \mu}^\ast} \oplus J^{cl}_{\mathfrak{p}}),\\
            \mathsf{m}_1\wedge \mathsf{m}_2 \wedge \ldots \wedge \mathsf{m}_n
            \otimes p &\mapsto \phi^{\mathsf{m}_1} \phi^{\mathsf{m}_2} \dots \phi^{\mathsf{m}_n}J^{cl}_p
        \end{aligned}
    \een
    for $\mathsf{m}\in \mathfrak{n}^*_{\lambda,\mu}$ and
    $p\in \mathfrak{p}$ is an isomorphism.
\end{proof}

Now using Proposition \ref{prop:Lie algebra coho-finite W-pre},
we can realize the generalized finite $W$-algebra $U(\lambda,\mu)$ via Lie algebra cohomology.

\begin{prop} \label{prop:Lie algebra coho-finite W}
    The Lie algebra cohomology of $U(\mathfrak{a})/ \mathcal{I}_{\lambda,\mu}$
    is concentrated on degree $0$ part; that is,
    \[
    \non
    H(\, \mathfrak{n}_{\lambda,\mu}  ,   U(\mathfrak{a})/ \mathcal{I}_{\lambda,\mu} \, )
    = H^0(\, \mathfrak{n}_{\lambda,\mu}, U(\mathfrak{a})/ \mathcal{I}_{\lambda,\mu}\, ).\]
    Hence $U(\lambda,\mu)\simeq H^0(\, \mathfrak{n}_{\lambda,\mu}  ,
    U(\mathfrak{a})/ \mathcal{I}_{\lambda,\mu} \, )=H(\, \mathfrak{n}_{\lambda,\mu}  ,
    U(\mathfrak{a})/ \mathcal{I}_{\lambda,\mu} \, ).$
\end{prop}
\begin{proof}
    The main idea is the same as for the proof of \cite[Proposition~5.2]{gg:qs}. Consider the
    cochain complex
        \[
        \non
        0 \longrightarrow  U(\mathfrak{a})/\mathcal{I}_{\lambda, \mu}
        \longrightarrow \mathfrak{n}_{\lambda, \mu}^\ast \otimes U(\mathfrak{a})/\mathcal{I}_{\lambda, \mu}
        \longrightarrow \dots \longrightarrow\wedge^n \mathfrak{n}_{\lambda, \mu}^\ast
        \otimes U(\mathfrak{a})/\mathcal{I}_{\lambda, \mu}  \longrightarrow \cdots\]
        and the  filtration on $\wedge^n \mathfrak{n}_{\lambda, \mu}^\ast
        \otimes U(\mathfrak{a})/\mathcal{I}_{\lambda, \mu}$ given by
        \begin{equation}
        \non
            F_p(\wedge^n \mathfrak{n}_{\lambda, \mu}^\ast
            \otimes U(\mathfrak{a})/\mathcal{I}_{\lambda, \mu})
            = \big\{ (E_{I_1}^\ast \wedge E_{I_2}^\ast \wedge \ldots
            \wedge E_{I_n}^\ast) \otimes v \ | \ \sum_{k=1}^n
            (\col_\mu(j_k) - \col_\mu(i_k))+j \leqslant p\big\}
        \end{equation}
        for $v \in K_jU(\mathfrak{a})/\mathcal{I}_{\lambda, \mu}$ and $I_k=(i_k,j_k,r_k)$.
        Considering the corresponding spectral sequence and its
        convergence, we get the proposition. The isomorphism
        $U(\lambda,\mu)\simeq H^0(\, \mathfrak{n}_{\lambda,\mu}  ,
        U(\mathfrak{a})/ \mathcal{I}_{\lambda,\mu} \, )$ follows directly
        from the definition of $U(\lambda,\mu).$
\end{proof}

\if
        Then the first page of the spectral sequence
        \ben
            E^{pq}_{0} = F_p(\wedge^{p+q} \mathfrak{n}_{\lambda, \mu}^- \otimes
            U(\mathfrak{a})/I_{\lambda, \mu}) / F_{p-1}
            (\wedge^{p+q} \mathfrak{n}_{\lambda, \mu}^- \otimes U(\mathfrak{a})/I_{\lambda, \mu}).
        \een
        Hence $
            E_{1}^{pq} = H^{p+q}(\mathfrak{n}_{\lambda, \mu},
            \text{gr}_p( U(\mathfrak{a})/I_{\lambda, \mu}))$
            and the spectral sequence converges to \[E_{\infty}^{p q}
            = F_pH^{p+q}(\mathfrak{n}_{\lambda, \mu},
            U(\mathfrak{a})/I_{\lambda, \mu}) / F_{p-1}H^{p+q}(\mathfrak{n}_{\lambda, \mu},
            U(\mathfrak{a})/I_{\lambda, \mu}),\]
            the second proposition holds.
\fi

Finally, we can prove the following theorem which is the main result in this section.

\begin{thm} \label{thm:zhu and finite}
    We have an isomorphism of associative algebras:
    \[
    Zhu_H(W^k(\lambda,\mu)) \cong U(\lambda,\mu).
    \]
\end{thm}
\begin{proof}
    By Proposition \ref{prop:Zhu and coho commute} and
    Proposition~\ref{prop:Lie algebra coho-finite W}, we know that
    \begin{equation}
    \non
        Zhu_H(W^k(\lambda, \mu)) \cong H(C^{\text{fin}}(\lambda, \mu), Q^{\text{fin}}),
        \qquad U(\lambda, \mu) = H(\mathfrak{n}_{\lambda, \mu}, U(\mathfrak{a})/\mathcal{I}_{\lambda, \mu}).
    \end{equation}
    Hence it suffices to show that the complexes
    $(C^{\text{fin}}(\lambda, \mu), Q^{\text{fin}})$
    and $(\wedge^\bullet \mathfrak{n}_{\lambda, \mu}^\ast
    \otimes U(\mathfrak{a})/\mathcal{I}_{\lambda, \mu}, d^c)$ are quasi-isomorphic, where
    \begin{equation}\label{eq:differentialofLAC}
    \begin{aligned}
        d^c(\Psi \otimes v) &= \frac{1}{2}
        \sum_{I \in S_{\lambda, \mu}} (E_{I})^\ast \wedge E_{I}
        \cdot \Psi \otimes v + \sum_{I \in S_{\lambda, \mu}}
        (E_{I})^\ast \wedge\Psi \otimes \text{ad}(E_{I})(v).
        \end{aligned}
    \end{equation}
    Here we set $E_I = E_{ij}^{(r)} = 0$ if $r \geqslant \lambda_j$ or
    $r < \lambda_j - \text{min}(\lambda_i, \lambda_j)$.
    Observe that
\ben
U(\mathfrak{a})/\mathcal{I}_{\lambda,\mu}
    \simeq U(\mathfrak{a})\otimes_{\mathfrak{n}_{\lambda,\mu}} \mathbb{C}_{-\chi},
\een
    where $\mathbb{C}_{-\chi}$ is the one-dimensional
    representation of $\mathfrak{n}_{\lambda,\mu}$ with $\mathsf{n}\cdot 1= -\chi(\mathsf{n})$.
    Hence we can follow the argument of \cite[Appendix]{dk:fa}.
\end{proof}

By Corollary \ref{cor:affine structure} and Theorem \ref{thm:zhu and finite},
we get the following corollary.

\begin{cor}  \label{cor:finite structure}
Let $a_1, \dots, a_r$ be a basis of $\text{ker}\, \varphi$ for the map $\varphi$ in \eqref{varphi}.
\begin{enumerate}[(1)]
    \item  As an associative algebra, $U(\lambda,\mu)$ has a generating set
    consisting of $r$ elements, where $r$ is the dimension of
    $\mathfrak{a}(0)$. Moreover, $r$ is the minimal number of elements in a generating set.
    \item Suppose $a_i\in K_{\Delta_i}(\mathfrak{a})\setminus K_{\Delta_i-1}(\mathfrak{a})$
    for $i=1,2,\dots, r.$ Let $v_i\in K_{\Delta_i}(\mathfrak{a})$ be an element
    in $U(\lambda,\mu)$ such that the linear term in $\text{gr}^K_{\Delta_i}(v_i)$ is $a_i.$
    Then the set $\{v_i\ |\ i=1,\dots, r\}$ generates $U(\lambda,\mu)$ and is algebraically independent.
\end{enumerate}
\end{cor}

\section{Principal and minimal nilpotent $\mu$}

Here we will discuss the $W$-algebras for two particular cases of $\mu$, while
keeping $\la$ arbitrary. By extending the terminology from the column-partition $\la$,
we will refer to the cases $\mu=(n)$ and $\mu=(1^{n-2}2)$ as the {\em principal}
and {\em minimal nilpotent} cases, respectively; cf. Remark~\ref{rem:krw}.

\subsection{Principal nilpotent case}

In this subsection, we describe generators of the generalized affine
and finite $W$-algebras when $\mu = (n)$.
In this particular case, the affine $W$-algebra $W^k(\lambda,(n))$
was introduced and described in \cite{m:wa}, although the finite
counterpart $U(\lambda, (n))$ was not discussed there.
We will prove here that the
center of $U(\mathfrak{a})$ is isomorphic to $U(\lambda, (n)).$

Keeping the definitions from Sec.~\ref{sec:finite}, note that
the dimension of $\mathfrak{a}(0)$ is the number of boxes in the pyramid $\la$,
which is equal to $N$.
Hence $\text{dim}(\text{ker} \, \varphi) = N$. We can
decompose \[\text{ker} \, \varphi=\bigoplus_{m=1}^n (\text{ker} \, \varphi)[m]\subset \mathfrak{p}\]
where $(\text{ker} \, \varphi)[m]$ is the subspace of $\text{ker}\, \varphi$
with conformal weight $m$. Then
\begin{equation} \label{eq:dim_conf_princ}
    \text{dim}(\text{ker} \, \varphi)[m]=\lambda_{n-m+1}
\end{equation}
since the set
\ben
E_{n, n-m+1}^{(r)} + E_{n-1, n-m}^{(r+\lambda_{n}-\lambda_{n-1})}
+ E_{n-2, n-m-1}^{(r+\lambda_{[n-1, n]}- \lambda_{[n-m, n-m+1]})}+
\dots + E_{m, 1}^{(r+\lambda_{[m+1, n]}-\lambda_{[2, n-m+1]})}
\een
with $0 \leqslant r \leqslant \lambda_{n-m+1}-1$
is a basis of $(\text{ker} \, \varphi)[m].$
Combining Corollary \ref{cor:affine structure} and \eqref{eq:dim_conf_princ},
we can derive following properties.

\begin{lem} \label{lem:princ_affine}
    Let $ A:=\{w_1,\dots,w_N\}$ be a subset of $W^k(\lambda,(n))$
    and $m=1,2,\dots, n$. If $A$ has $\lambda_{n-m+1}$ elements whose
    linear parts without total derivatives span a subspace in $\mathfrak{p}$
    of dimension $\lambda_{n-m+1}$ with conformal weight $m$,
    then $A$
    freely generates $W^k(\lambda,(n))$ as a differential algebra.
\end{lem}

Note that this agrees with \cite{m:wa}, where
generators of $W^k(\lambda,(n))$ were described.

\begin{ex} \label{rem:no conformal vector}
    Let $\lambda = (2, 2)$. Recall from \eqref{eq:affine and W} that $W^k(\lambda,(n))$
    can be regarded as a vertex subalgebra of $V(J_{\mathfrak{p}}).$
    According to \cite{m:wa}, conformal weight $1$ generators of $W^k(\lambda, (n))$ are
    \begin{equation}
    \non
        w_1 = J_{E_{11}^{(0)}} + J_{E_{22}^{(0)}}, \quad w_2 = J_{E_{11}^{(1)}} + J_{E_{22}^{(1)}}
    \end{equation}
    and conformal weight $2$ generators of $W^k(\lambda, (n))$ are
    \begin{equation}
    \non
        w_3 = J_{E_{21}^{(0)}} + :J_{E_{11}^{(0)}} J_{E_{22}^{(1)}}:
        + :J_{E_{11}^{(1)}} J_{E_{22}^{(0)}}: - (k+2) \partial J_{E_{11}^{(1)}}, \quad w_4
        = J_{E_{21}^{(1)}} + :J_{E_{11}^{(1)}} J_{E_{22}^{(1)}} : \, .
    \end{equation}
    The $\lambda$-brackets between generators are
    \begin{equation}
    \non
        [w_3 {}_\lambda w_3] = (\partial + 2\lambda)\left(-(k+4) w_4
        + \left( \frac{k}{4} + 1 \right) :w_2w_2:\right)
    \end{equation}
    and all other $\lambda$-brackets are $0$. Note that, in this example,
    we can see that $W^k(\lambda,(n))$ does not have a conformal vector; see Remark \ref{rem:conf}.
    If we look for a conformal vector of weight $2$ in the form
    \begin{equation}
    \non
        w := c_1 \partial w_1 + c_2 \partial w_2 + c_{11} :w_1 w_1:
        + c_{12} : w_1 w_2: + c_{22} :w_2 w_2: + c_3 w_3 + c_4 w_4
    \end{equation}
    and solve the equation
    \begin{equation}
    \non
        [w {}_\lambda w] = (\partial + 2\lambda) w + c\lambda^3
    \end{equation}
    then there will be no non-trivial solution $(c_1, c_2, c_{11}, c_{12}, c_{22}, c_3, c_4)$.
\end{ex}

Now we want to find a generating set of $U(\lambda,(n))$. Combining again
Corollary \ref{cor:finite structure} and \eqref{eq:dim_conf_princ}, we get the following lemma.

\begin{lem}\label{lem:princ_finite}
    Let $F:=\{v_1,\dots,v_N\}$ be a subset of $U(\lambda,(n))$
    and $m=1,2,\dots, n$. Suppose $F$ has $\lambda_{n-m+1}$ elements
    in $K_m(\mathfrak{a})\setminus K_{m-1}(\mathfrak{a})$
    whose linear parts in the image by $\text{gr}^K_{m}$ span
    a subspace in $\mathfrak{p}$ of
dimension $\lambda_{n-m+1}$,  where $K_{\bullet}$
    is the Kazhdan filtration. Then $F$ freely generates $U(\lambda,(n))$ as an associative algebra.
\end{lem}

To find a generating set of $U(\lambda,(n))$, recall
a result from \cite{m:cc}. Let
\begin{equation}
\non
    \mathcal{M}:= \left[\begin{array}{cccc} x +(n-1)\lambda_1 +\epsilon_{11}(u) &
    \epsilon_{12}(u) & \cdots & \epsilon_{1n}(u) \\
    \epsilon_{21}(u) & x+(n-2) \lambda_2 + \epsilon_{22}(u) & \cdots & \epsilon_{2n}(u) \\
    \vdots & \vdots & \ddots & \vdots \\
    \epsilon_{n1}(u) & \epsilon_{n2}(u) & \cdots & x+\epsilon_{nn}(u)
    \end{array} \right]
\end{equation}
be the matrix with entries in $\mathbb{C}[x]\otimes U(\mathfrak{a})[u]$, where
\begin{equation}
\non
    \epsilon_{ij}(u)= \left\{\begin{array}{ll}
        E_{ij}^{(0)}+ E_{ij}^{(1)}u+\dots + E_{ij}^{(\lambda_j-1)} u^{\lambda_j-1} &
        \text{ if } i\geqslant  j,  \\[0.4em]
        E_{ij}^{(\lambda_j-\lambda_i)} u^{\lambda_j-\lambda_i}+\dots
        + E_{ij}^{(\lambda_j-1)} u^{\lambda_j-1}&   \text{ if } i< j.
    \end{array} \right.
\end{equation}
Let $\Phi_m(u)=\sum_{r\in \mathbb{Z}_+}\Phi_m^{(r)}u^r$ for  $\Phi_m^{(r)}\in U(\mathfrak{a})$, be
the coefficients in the expansion of
the column determinant of $\mathcal{M}$,
\ben
    \text{cdet}\ts\mathcal{M}= x^n+ \Phi_1(u)x^{n-1}+ \dots +\Phi_n(u).
\een
Then the following is an algebraically independent
generating set of the center of the algebra $U(\mathfrak{a})$ \cite[Cor.~2.7]{m:cc}:
\begin{equation}\label{eq:gen_center}
    \{\, \Phi_m^{(r)}\, | \, \lambda_{n-m+2}+\lambda_{n-m+3}+\dots
    + \lambda_n < r+m\leqslant \lambda_{n-m+1}+\lambda_{n-m+2}+\dots + \lambda_n\, \}.
\end{equation}

\begin{thm} \label{thm:principal case}
Let $m=1,2,\dots,n$ and $\Phi_m^{(r)}\in U(\mathfrak{a})$ be elements in \eqref{eq:gen_center}.
Let $\Psi_m^{(r)}$ be the image of $\Phi_m^{(r)}$ in $U(\mathfrak{a})/\mathcal{I}_{\lambda,\mu}.$ Then
\begin{equation}\label{eq:finite W}
    \{\, \Psi_m^{(r)}\, | \, \lambda_{n-m+2}+\lambda_{n-m+3}+\dots + \lambda_n < r+m\leqslant
    \lambda_{n-m+1}+\lambda_{n-m+2}+\dots + \lambda_n\, \}
\end{equation}
is an algebraically independent generating set of $U(\lambda,(n))$.
\end{thm}

\begin{proof}
    It is obvious that $\Psi_m^{(r)}\in U(\lambda,(n))$ since $[\mathsf{n}, \Phi_m^{(r)}]=0$
    for any $\mathsf{n}\in \mathfrak{n}_{\lambda,(n)}.$ Also, it is not hard to check that
    $\Psi_m^{(r)}$ is nonzero element in the quotient
    \ben
    \big(K_m(\mathfrak{a})\setminus K_{m-1}(\mathfrak{a})\big)
    \big/ \mathcal{I}_{\lambda, \mu}\big(K_m(\mathfrak{a})\setminus K_{m-1}(\mathfrak{a})\big).
    \een
    Hence $\text{gr}^K_m(\Psi_m^{(r)})$ is nontrivial.
    Set
    \[
    \non
    l_m(u):=\sum_{r\in \mathbb{Z}_+} l_m^{(r)} u^r,\] where
    $l_m^{(r)}$ is the linear part of $\text{gr}^K_m(\Psi_{m}^{(r)}).$
    By Lemma \ref{lem:princ_finite}, it is enough to show that
    \begin{equation} \label{eq:linear part}
          \{\, l_m^{(r)}\, | \, \lambda_{n-m+2}+\lambda_{n-m+3}+\dots
          + \lambda_n < r+m\leqslant \lambda_{n-m+1}+\lambda_{n-m+2}+\dots + \lambda_n\, \}
    \end{equation}
    is a linearly independent subset in $\mathfrak{p}.$ By direct computations, we get
    \begin{equation} \label{eq:linear}
        \begin{aligned}
             l_m(u)&= \epsilon_{n, n-m+1}(u)
             + \epsilon_{n-1, n-m+2}(u) u^{\lambda_n-\lambda_{n-m+1}} \\
             &+ \epsilon_{n-2, n-m-1}(u) u^{\lambda_n
             + \lambda_{n-1} -\lambda_{n-m+1}-\lambda_{n-m}}
             + \dots + \epsilon_{m, 1}(u) u^{\lambda_{[m+1, n]} - \lambda_{[2, n-m+1]}}
        \end{aligned}
    \end{equation}
    for $\lambda_{[s,t]}:=\lambda_s+\lambda_{s+1}+\dots+\lambda_t$
   and hence
   \ben
       l_m^{(r)}=E_{n, n-m+1}^{(r)} + E_{n-1, n-m+2}^{(r-\lambda_n+\lambda_{n-m+1})}
       + \dots + E_{m1}^{(r-\lambda_{[m+1, n]} + \lambda_{[2, n-m+1]})}.
   \een
   This implies that the set \eqref{eq:linear part} is linearly independent, completing the proof.
\end{proof}

\begin{ex}
    Let $n=3.$ Let us verify that \eqref{eq:finite W} is an algebraically independent
    generating set of $U(\lambda,(3)).$ In order to find the image of
    $\Psi_m^{(r)}$ in the graded algebra with respect to the Kazhdan filtration,
    consider the following matrix:
    \begin{equation} \label{eq:ex 3by3}
        \overline{\mathcal{M}}_3=\left[\begin{array}{ccc}
        x+2\lambda_1+\epsilon_{11}(u) & -u^{\lambda_2-1} & 0 \\
        \epsilon_{21}(u) & x+ \lambda_2+\epsilon_{22}(u) & -u^{\lambda_3-1} \\
        \epsilon_{31}
        (u) & \epsilon_{32}(u) & x+\epsilon_{33}(u)
        \end{array}\right].
    \end{equation}
    We substituted $\epsilon_{12}(u)$, $\epsilon_{23}(u)$
    and  $\epsilon_{13}(u)$ by their respective images $-u^{\lambda_2-1}$, $-u^{\lambda_3-1}$
    and $0$ in $U(\mathfrak{a})/\mathcal{I}_{\lambda,\mu}.$
    Since the graded algebra of $U(\mathfrak{a})/\mathcal{I}_{\lambda,\mu}$
    is commutative, we can find $\text{gr}(\Psi_m^{(r)})$ by computing the column
    determinant of \eqref{eq:ex 3by3} which is given by
    \begin{equation}
    \non
    \begin{aligned}
        & \text{\rm cdet}\ \overline{\mathcal{M}}_3 = x^3+ x^2 \Big(\epsilon_{11}
        +\epsilon_{22}+\epsilon_{33}+2\lambda_1+\lambda_2\Big)\\
        & +x \Big( (2\lambda_1+\epsilon_{11})(\lambda_2+\epsilon_{22})
        +(2\lambda_1+\epsilon_{11})\epsilon_{33}+(\lambda_2+\epsilon_{22})\epsilon_{33}
        +\epsilon_{32}u^{\lambda_3-1}+\epsilon_{21}u^{\lambda_2-1} \Big)\\
        & + \Big( (2\lambda_1+\epsilon_{11})(\lambda_2+\epsilon_{22})\epsilon_{33}
        +(2\lambda_1+\epsilon_{11})\epsilon_{32}u^{\lambda_3-1}
        +\epsilon_{21}\epsilon_{33}u^{\lambda_2-1}+\epsilon_{31}u^{\lambda_2+\lambda_3-2}\Big).
    \end{aligned}
    \end{equation}
    By setting $\Psi_{m}(u):= \sum_{r\in \mathbb{Z}_+}\Psi_{m}^{(r)}u^r$ and
    $\text{gr}\, \Psi_{m}(u):= \sum_{r\in \mathbb{Z}_+}\text{gr}(\Psi_{m}^{(r)})u^r$ we find
    \begin{equation}
    \non
        \begin{aligned}
           &  \text{gr}\, \Psi_{1}(u)= \epsilon_{11}+\epsilon_{22}+\epsilon_{33}+2\lambda_1+\lambda_2,\\
           & \text{gr}\, \Psi_{2}(u)= (2\lambda_1+\epsilon_{11})(\lambda_2
           +\epsilon_{22})+(2\lambda_1+\epsilon_{11})\epsilon_{33}
           +(\lambda_2+\epsilon_{22})\epsilon_{33}+\epsilon_{32}u^{\lambda_3-1}
           +\epsilon_{21}u^{\lambda_2-1},\\
           & \text{gr}\, \Psi_{3}(u)=(2\lambda_1+\epsilon_{11})(\lambda_2
           +\epsilon_{22})\epsilon_{33}+(2\lambda_1+\epsilon_{11})
           \epsilon_{32}u^{\lambda_3-1}+\epsilon_{21}\epsilon_{33}
           u^{\lambda_2-1}+\epsilon_{31}u^{\lambda_2+\lambda_3-2}.
        \end{aligned}
    \end{equation}
    Hence $l_m(u)$ in \eqref{eq:linear} for $m=1,2,3$ are
    \begin{equation}
    \non
        \begin{aligned}
            & l_1= \epsilon_{11}+\epsilon_{22}+\epsilon_{33}, \quad l_2
            =\epsilon_{32}+\epsilon_{21}u^{\lambda_3-\lambda_2},\quad  l_3=\epsilon_{31}\ .
        \end{aligned}
    \end{equation}
    Now, by direct computations, one can check that $l_1^{(r_1)}$, $l_2^{(r_2)}$,
    $l_3^{(r_3)}$ for $r_1\in \{ 0,1,\dots, \lambda_3-1\}$,
    $r_2\in \{\lambda_3-1, \dots, \lambda_2+\lambda_3-2\}$,
    $r_3\in \{\lambda_2+\lambda_3-2, \dots, \lambda_1+\lambda_2+\lambda_3-3\}$ are linearly independent.
\end{ex}

As a direct consequence of Theorem \ref{thm:principal case}, we get the following corollary.

\begin{cor}\label{cor:cent}
    The generalized finite $W$-algebra $U(\lambda,(n))$ is
    isomorphic to the center of $U(\mathfrak{a})$ as an associative algebra.
    In particular, the algebra $U(\lambda,(n))$ is commutative.
\end{cor}

\subsection{Minimal nilpotent case}
\label{subsec:minnilp}
In this subsection, we describe generators of generalized finite
and affine $W$-algebras when $\lambda$ is arbitrary and
$\mu = (1, 1, \dots, 1, 2)$ (with $n-2$ parts equal to $1$).
From the $\mathbb{Z}$-grading \eqref{eq:deg_mu},
\begin{equation}
\non
    \mathfrak{n}_{\lambda, \mu}
    = \text{Span}_{\mathbb{C}} \{ E_{\beta n}^{(r)} \ | \ 1 \leqslant \beta \leqslant n-1,
    \lambda_n-\lambda_\beta \leqslant r < \lambda_n\}.
\end{equation}
To apply Corollaries~\ref{cor:affine structure} and \ref{cor:finite structure},
we first describe the kernel of $\varphi$. Since $\chi = (E_{n-1, n}^{(\lambda_n-1)})^\ast$,
we have
\begin{equation}
\non
    \text{ker} \, \varphi = \text{Span}_{\mathbb{C}}\left( \mathcal{B}_1 \cup \mathcal{B}_2  \right)
\end{equation}
where
\begin{equation}
\non
\begin{aligned}
    \mathcal{B}_1 &:= \{ E_{ij}^{(r)} \ | \ 1 \leqslant i \leqslant n-2, 1 \leqslant j \leqslant n-1,
    \lambda_j-\text{min}(\lambda_i, \lambda_j)\leqslant r <\lambda_j\} \\
    & \hskip 6cm \cup \{E_{n-1, n-1}^{(r)} + E_{nn}^{(r)} \ | \ 0 \leqslant r < \lambda_n - 1\}, \\[0.4em]
    \mathcal{B}_2 &:= \{ E_{n \alpha }^{(r)} \ | \ 1 \leqslant \alpha \leqslant n-1,
    0 \leqslant r < \lambda_\alpha-1\}.
    \end{aligned}
\end{equation}
Therefore the number of free generators of the generalized finite and affine $W$-algebras is
\begin{equation}
\non
    \text{dim} \, \mathfrak{a}(0)
    = \text{dim} \, (\text{ker} \, \varphi) = (2n-3)\lambda_1 + (2n-5)\lambda_2
    + \dots + \lambda_{n-1} + \lambda_n \ .
\end{equation}
Moreover, the numbers of generators of conformal weights $1$ and $2$ are
\ben
|S_1|
= (2n-4)\lambda_1 + (2n-6)\lambda_2 + \dots + 2\lambda_{n-2}
+ \lambda_n\Fand |S_2| = \lambda_1 + \lambda_2 + \dots + \lambda_{n-1},
\een
respectively.
In the following theorems, we find generating sets of $U(\lambda,\mu)$ and $W^k(\lambda,\mu).$

\begin{thm} \label{thm:finite generators}
The generalized finite $W$-algebra $U(\lambda,\mu)$ has the following properties:
    \begin{enumerate}[(1)]
        \item  Any element $a \in \mathcal{B}_1$ is in $U(\lambda,\mu)$
        and has conformal weight 1.
        \item  For  $E_{n\alpha}^{(r)} \in \mathcal{B}_2$, the element
            \begin{equation} \label{eq:finite min weight 2}
                E_{n\alpha}^{(r)} + \sum_{a+b=\lambda_n-1+r} E_{n-1, \alpha}^{(a)} E_{nn}^{(b)}
                - \sum_{\gamma=1}^{n-2}\ \sum_{a+b = \lambda_n -1 + r}
                E_{n-1, \gamma}^{(a)} E_{\gamma \alpha}^{(b)}
                - \delta_{r, 0} \delta_{\lambda_\alpha, \lambda_n}
                \lambda_n E_{n-1, \alpha}^{(\lambda_n-1)}
            \end{equation}
        is in $U(\lambda,\mu)$ and has conformal weight $2$.
        \item The set $\mathcal{B}_1$ and the elements in \eqref{eq:finite min weight 2}
        comprise an algebraically independent generating set of $U(\lambda,\mu).$
    \end{enumerate}
\end{thm}

The proof of Theorem \ref{thm:finite generators} is a simple analogue
of the proof of the next theorem. Recall the property of $W^k(\lambda,\mu)$,
established in \eqref{eq:affine and W}.

\begin{thm} \label{thm:affine generators}
  For the generators
  of $W^k(\lambda,\mu),$ we have the following properties.
    \begin{enumerate}[(1)]
        \item For $a \in \mathcal{B}_1$, the element $J_a$ is in $W^k(\lambda,\mu)$
        and has conformal weight 1.
        \item For $E_{n\alpha}^{(r)} \in \mathcal{B}_2$, the element
                \begin{equation} \label{eq:minimal affine generator}
                \begin{aligned}
            G_{n\alpha}^{(r)} := J_{E_{n\alpha}^{(r)}} &
            + \sum_{a+b= \lambda_n - 1 +r}  :J_{E_{n-1, \alpha}^{(a)}}J_{E_{nn}^{(b)}}:
            - \sum_{\gamma=1}^{n-2}\
            \sum_{a+b = \lambda_n -1 + r} :J_{E_{n-1, \gamma}^{(a)}} J_{E_{\gamma \alpha}^{(b)}}:\\
            &- \delta_{r, 0} \delta_{\lambda_\alpha, \lambda_n}
            \left(\lambda_n +k \right)\partial J_{E_{n-1, \alpha}^{(\lambda_n-1)}}.
            \end{aligned}
        \end{equation}
        is in $W^k(\lambda, \mu)$ and has conformal weight $2$.
        \item The set $\{J_a \ | \ a \in \mathcal{B}_1\}$ and elements in
        \eqref{eq:minimal affine generator} comprise a differential
        algebraically independent generating set of $W^k(\lambda, \mu)$.
    \end{enumerate}
\end{thm}

\begin{proof}
\begin{enumerate}[(1)]
   \item  By Theorem \ref{thm:degree zero}, it is sufficient to show
   that the above elements are inside $\widetilde{Q}(J_a)=0$
   for $a\in \mathcal{B}_1$. Let $1\leqslant i \leqslant n-2$ and $1 \leqslant j \leqslant n-1$.
   Then
    \begin{equation}
    \non
        \widetilde{Q}(J_{E_{ij}^{(r)}})
        = -\delta_{jn} \phi^{E_{n-1, i}^{(\lambda_n-1-r)}{}^\ast}
        + \delta_{i, n-1} \phi^{E_{jn}^{(\lambda_n-1-r)}{}^\ast} = 0
    \end{equation}
    and
    \begin{equation}
    \non
        \widetilde{Q}(J_{E_{n-1, n-1}^{(r)} + E_{nn}^{(r)}})
        = -\phi^{E_{n-1, n}^{(\lambda_n -1 -r)} {}^\ast}
        + \phi^{E_{n-1, n}^{(\lambda_n -1-r)} {}^\ast} = 0 \, .
    \end{equation}
    Hence we get the assertion.
    \item
    Now suppose $E_{n\alpha}^{(r)} \in \mathcal{B}_2$. Then
    \begin{equation}\non
        \begin{aligned}
            \widetilde{Q}(J_{E_{n\alpha}^{(r)}}) &
            = \sum_{\gamma=1}^{n-1}\sum_{a,b} : \phi^{E_{\gamma n}^{(\lambda_n-1-a)}}{}^\ast
            J_{E_{\gamma \alpha}^{(b)}} :
            - \sum_{a,b} :\phi^{E_{\alpha n}^{(\lambda_n -1-a)}}{}^\ast J_{E_{nn}^{(b)}} :\\
            &\qquad + \delta_{\lambda_\alpha, \lambda_n} \delta_{r, 0}
            \frac{k}{N} (\lambda_1 + \dots + \lambda_{\alpha-1}
            + (n-\alpha+1)\lambda_\alpha)\partial \phi^{E_{\alpha n}^{(0)}{}^\ast}
          \end{aligned}
            \end{equation}
            which equals
          \begin{equation}\non
        \begin{aligned}
            &\phantom{=} \sum_{\gamma=1}^{n-1}\sum_{a,b} :
            \phi^{E_{\gamma n}^{(\lambda_n-1-a)}}{}^\ast J_{E_{\gamma \alpha}^{(b)}} :
            - \sum_{a,b} :\phi^{E_{\alpha n}^{(\lambda_n -1-a)}}{}^\ast J_{E_{nn}^{(b)}} :\\
            &\qquad + \delta_{\lambda_\alpha, \lambda_n} \delta_{r, 0}
            \frac{k}{N} (\lambda_1 + \dots + \lambda_{\alpha -1}
            + \lambda_\alpha + \lambda_{\alpha+1} + \dots + \lambda_n )
            \partial  \phi^{E_{\alpha n}^{(0)}}{}^\ast \\
            &= \sum_{\gamma=1}^{n-1}\sum_{a,b} :
            \phi^{E_{\gamma n}^{(\lambda_n-1-a)}}{}^\ast J_{E_{\gamma \alpha}^{(b)}} :
            - \sum_{a,b} :\phi^{E_{\alpha n}^{(\lambda_n -1-a)}}{}^\ast J_{E_{nn}^{(b)}} :
            + \delta_{\lambda_\alpha, \lambda_n} \delta_{r, 0} k\partial
            \phi^{E_{\alpha n}^{(0)}}{}^\ast,
            \end{aligned}
            \end{equation}
            where all the summation for $a,b$ runs over the all the
            integer values satisfying $a+b=\lambda_n-1+r.$
            One the other hand, when $a+b=\lambda_n-1+r$ and $\gamma=1,\dots, n-2$, we have
            \begin{align}
                &   \widetilde{Q}(  :J_{E_{n-1, \alpha}^{(a)}} J_{E_{nn}^{(b)}}:)
                =   :\phi^{E_{\alpha n}^{(\lambda_n -1 -a)}{}^\ast} J_{E_{nn}^{(b)}} :
                - :\phi^{E_{n-1, n}^{(\lambda_n-1-a)}}{}^\ast J_{E_{n-1, \alpha}^{(b)}} :
                +\delta_{\lambda_\alpha, \lambda_n} \delta_{r, 0}
                \lambda_n \partial \phi^{E_{\alpha n}^{(0)}{}^\ast},  \non \\
                & \widetilde{Q}( :J_{E_{n-1, \gamma}^{(a)}} J_{E_{\gamma \alpha}^{(b)}}:)
                =  :\phi^{E_{\gamma n}^{(\lambda_n-1-a)}} {}^\ast
                J_{E_{\gamma \alpha}^{(b)}} : , \non \\
                & \widetilde{Q}(\partial J_{E_{n-1, \alpha}^{(\lambda_n-1)}})
                = \delta_{r, 0} \delta_{\lambda_n, \lambda_\alpha}
                \partial \phi^{E_{\alpha n}^{(0)}{}^\ast} .\non
            \end{align}
            These relations imply that $\widetilde{Q}(G_{n\alpha}^{(r)})=0$
        and so this element is in $W^k(\lambda,\mu).$
  \item This follows directly from Corollary \ref{cor:affine structure}.
        \end{enumerate}
\end{proof}

\begin{ex}
    Let $\lambda = (1, 1, 2, 2)$ and $\mu = (1, 1, 2)$. Then
    \begin{equation}
    \non
    \begin{aligned}
        \mathcal{B}_1 &= \{E_{11}^{(0)}, E_{12}^{(0)}, E_{13}^{(1)},
        E_{21}^{(0)}, E_{22}^{(0)}, E_{23}^{(1)}\} \cup \{E_{33}^{(0)}
        + E_{44}^{(0)}, E_{33}^{(1)} + E_{44}^{(1)} \} \ , \\
        \mathcal{B}_2 &= \{E_{41}^{(0)}, E_{42}^{(0)},E_{43}^{(1)},  E_{43}^{(0)} \} \ .
    \end{aligned}
    \end{equation}
    For $a \in \mathcal{B}_1$, $a$ and $J_{a}$ are generators of
    $U(\lambda, \mu)$ and $W^k(\lambda, \mu)$, respectively, of conformal
    weight~$1$.
    The generators of $U(\lambda, \mu)$ of conformal weight $2$ are
    \begin{equation}
    \non
        \begin{aligned}
            & E_{41}^{(0)} + E_{31}^{(0)} E_{44}^{(1)}\ ,\quad
            E_{42}^{(0)} + E_{32}^{(0)} E_{44}^{(1)}\ , \quad E_{43}^{(1)} + E_{33}^{(1)} E_{44}^{(1)}\ , \\
            & E_{43}^{(0)} + E_{33}^{(0)} E_{44}^{(1)}
            + E_{33}^{(1)} E_{44}^{(0)} -E_{31}^{(0)}E_{13}^{(1)} - E_{32}^{(0)} E_{23}^{(1)}-2 E_{33}^{(1)}
        \end{aligned}
    \end{equation}
    and the generators of  $W^k(\lambda, \mu)$ of conformal weight $2$ are
    \begin{equation}
    \non
        \begin{aligned}
            & J_{E_{41}^{(0)}} + :J_{E_{31}^{(0)}} J_{E_{44}^{(1)}}:\ , \quad
            J_{E_{42}^{(0)}} + :J_{E_{32}^{(0)}} J_{E_{44}^{(1)}}:\ ,\quad
            J_{E_{43}^{(1)}} + :J_{E_{33}^{(1)}} J_{E_{44}^{(1)}}:\ ,\\
            & J_{E_{43}^{(0)}} + :J_{E_{33}^{(0)}} J_{E_{44}^{(1)}}: + :J_{E_{33}^{(1)}}
            J_{E_{44}^{(0)}}: -:J_{E_{31}^{(0)}} J_{E_{13}^{(1)}}:
            - :J_{E_{32}^{(0)}} J_{E_{23}^{(1)}}: -(k+2)\partial J_{E_{33}^{(1)}}\ . \\
        \end{aligned}
    \end{equation}
\end{ex}

\begin{ex}
    If $\lambda= (1, 1, 2, 3)$ and $\mu = (1, 1, 2)$, then the last terms of
    \eqref{eq:finite min weight 2} and \eqref{eq:minimal affine generator}
    are equal to $0$. In this case, the generators of $W^k(\lambda, \mu)$
    of conformal weight $2$ are
    \begin{equation}
    \non
        \begin{aligned}
            J_{E_{41}^{(0)}} &+ :J_{E_{31}^{(0)}}J_{E_{44}^{(2)}}:\ ,
            \quad J_{E_{42}^{(0)}} + :J_{E_{32}^{(0)}}J_{E_{44}^{(2)}}: \\
            J_{E_{43}^{(1)}} & + :J_{E_{33}^{(1)}} J_{E_{44}^{(2)}}:\ ,
            \quad J_{E_{43}^{(0)}} + :J_{E_{33}^{(0)}}J_{E_{44}^{(2)}}:\ .
        \end{aligned}
    \end{equation}
\end{ex}

\section*{Declarations}

\subsection*{Competing interests}
The authors have no competing interests to declare that are relevant to the content of this article.

\subsection*{Acknowledgements}
This project was initiated during the second named author visit to the Seoul National University. He
is grateful to the Department of Mathematical Sciences for the warm hospitality.
His work was also supported by the Australian Research Council, grant DP240101572.
The work of U.R.~Suh was supported by NRF Grant \#2022R1C1C1008698
and  Creative-Pioneering Researchers Program by Seoul National University.

\subsection*{Availability of data and materials}
No data was used for the research described in the article.



\bigskip
\bigskip

\small

\noindent
D.J.C.:\newline
Department of Mathematical Sciences\\
Seoul National University, Gwanak-ro 1, Gwanak-gu, Seoul 08826, Korea\\
{\tt djchoi9696@snu.ac.kr}

\vspace{5 mm}

\noindent
A.M.:\newline
School of Mathematics and Statistics\newline
University of Sydney,
NSW 2006, Australia\newline
{\tt alexander.molev@sydney.edu.au}

\vspace{5 mm}

\noindent
U.R.S:\newline
Department of Mathematical Sciences and Research Institute of Mathematics\\
Seoul National University, Gwanak-ro 1, Gwanak-gu, Seoul 08826, Korea\\
{\tt uhrisu1@snu.ac.kr}

\end{document}